%% file: main.tex
\documentclass[a4paper,12pt]{article}

\input{header}
\usepackage[normalem]{ulem}

\title{Generative modeling with Gaussian Boson Sampling: classically trainable Bosonic Born Machines}

\author[1,2]{Zolt\'an Kolarovszki}
\author[1,2]{Bence Bak\'o}
\author[3]{Micha{\l} Oszmaniec}
\author[4]{Changhun Oh}
\author[1,5,6]{Zolt\'an Zimbor\'as}
\affil[1]{HUN-REN Wigner Research Centre for Physics, Budapest, Hungary}
\affil[2]{E\"otv\"os Lor\'and University,  Budapest, Hungary}
\affil[3]{Center for Quantum Enabled-Computing, Center for Theoretical Physics of
the Polish Academy of Sciences, Al. Lotników 32/46, 02-668 Warsaw, Poland}
\affil[4]{Department of Physics, Korea Advanced Institute of Science and Technology, Daejeon 34141, Korea}
\affil[5]{University of Helsinki, Yliopistonkatu 4 00100 Helsinki, Finland}
\affil[6]{Algorithmiq Ltd, Kanavakatu 3C 00160 Helsinki, Finland}

\date{}

\usepackage{amsfonts,amsmath,amssymb,physics,mathtools,bbm}

\begin{document}

\maketitle

\begin{abstract}
    Quantum generative modeling has emerged as a promising application of quantum computers, aiming to model complex probability distributions beyond the reach of classical methods. In practice, however, training such models often requires costly gradient estimation performed directly on the quantum hardware. Crucially, for certain structured quantum circuits, expectation values of local observables can be efficiently evaluated on a classical computer, enabling classical training without calls to the quantum hardware in the optimization loop. In these models, sampling from the resulting circuits can still be classically hard, so inference must be performed on a quantum device, yielding a potential computational advantage. In this work, we introduce a photonic quantum generative model built on parametrized Gaussian Boson Sampling circuits. The training is based on the efficient classical evaluation of expectation values enabled by the Gaussian structure of the state, allowing scalable optimization of the model parameters through the maximum mean discrepancy loss function. We demonstrate the effectiveness of the approach through numerical experiments on photonic systems with up to 805 modes and over a million trainable parameters, highlighting its scalability and suitability for near-term photonic quantum devices.
\end{abstract}

\section{Introduction}
    The success of classical deep learning is closely linked to the favorable trainability of deep neural networks, enabled by efficient gradient computation techniques that support optimization in high-dimensional parameter spaces~\cite{Goodfellow-et-al-2016,baydin2018automatic,bengio2014deep}. In contrast, quantum machine learning (QML) models face many challenges, including unfavorable loss landscape~\cite{anschuetz2022quantum,mcclean2018barren,larocca2025barren}, lack of automatic differentiation via backpropagation~\cite{10.5555/3666122.3668062,gilyen2019optimizing} and significant shot-noise in loss function evaluations~\cite{hoefler2023disentanglinghypepracticalityrealistically}. 
    More concretely, unlike classical neural networks, where gradients can often be computed at a cost comparable to loss evaluation, standard gradient-estimation techniques for parametrized quantum circuits require a number of circuit evaluations that scales with the number of trainable parameters~\cite{mitarai2018,schuld2019,Wierichs_2022,Pappalardo_2025}. For large models, this leads to a substantial sampling overhead, which is further exacerbated by the comparatively low sampling rates of current quantum hardware~\cite{hoefler2023disentanglinghypepracticalityrealistically}, the inherently probabilistic nature of quantum measurements, and barren plateaus~\cite{mcclean2018barren,larocca2025barren}.
    
    Proposed strategies for overcoming such challenges include developing more efficient gradient estimation methods~\cite{Sweke_2020,coyle2025training,bowles2025backpropagation} and restricting circuit architectures to improve loss landscapes~\cite{PRXQuantum.5.020328,Monbroussou2025trainability,PRXQuantum.3.030341,PRXQuantum.4.010328,PRXQuantum.4.020327}.
    However, the conditions that ensure good trainability in quantum models by avoiding barren plateaus often align with the ability to efficiently simulate expectation values of local observables using classical methods. This presents a challenge for supervised quantum machine learning models, as they typically output such expectation values, casting doubt on whether they can achieve true quantum utility. In contrast, generative learning frameworks are not constrained by this limitation, making them particularly promising candidates for demonstrating quantum computational advantage.

    Early attempts at classically training quantum generative models relied on estimating global output probabilities~\cite{kasture2023protocols}, which limits scalability. Building on trainability results for generative models, the authors of Ref.~\cite{rudolph2024trainability} proposed an expectation-value-based formulation of the squared maximum mean discrepancy (\mmd) loss function that was later used to construct a scalable training framework for parametrized IQP circuits~\cite{recio2025iqpopt,recio2025train}. Subsequent work established universality of such models under additional resources~\cite{kurkin2025universality}. More recently, an alternative architecture, called Fermionic Born Machine~\cite{bakó2025fermionicbornmachinesclassical}, was introduced, providing a classically trainable generative model based on the Fermion Sampling quantum advantage scheme~\cite{oszmaniec2022fermion}.
    
    In this work, we take a step away from qubit-based QML models and introduce \emph{Gaussian Bosonic Born Machines} (GBBMs), which form a class of classically trainable quantum generative models inspired by Gaussian Boson Sampling~\cite{hamilton2017,Kruse_2019}. Our approach primarily employs parity-operator strings as observables, which play a role analogous to Pauli-$Z$ operators in qubit-based models and enable an expectation-value formulation of the $\mathrm{MMD}^2$ loss in the bosonic setting. Crucially, this allows the loss to be evaluated efficiently using only the first and second moments of the Gaussian state, avoiding explicit sampling during training, while still retaining sampling hardness in the appropriate regimes~\cite{hamilton2017,go2025sufficientconditionshardnesslossy,li2025complexity_transition_dgbs}. The overall training and inference workflow of the GBBM is depicted in \cref{fig:gbbm_schematic}.

    \begin{figure}
        \centering
        \includegraphics[width=0.7\linewidth]{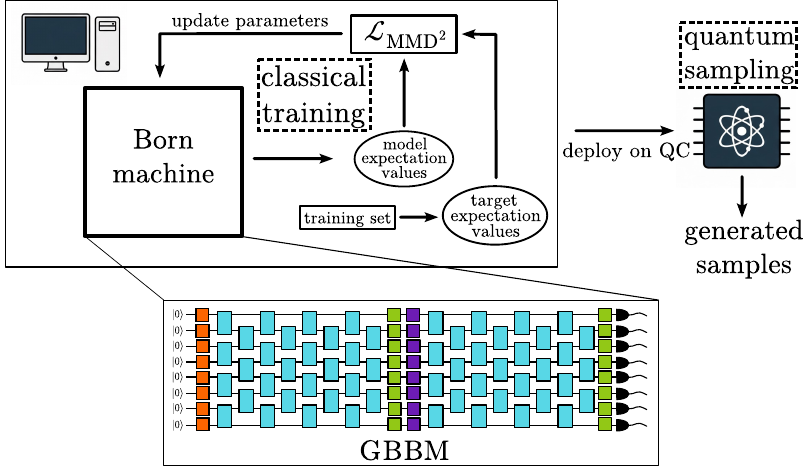}
        \caption{Schematic overview of the Gaussian Bosonic Born Machine (GBBM) framework. During \emph{classical training}, the Born machine is optimized by minimizing an $\mathcal{L}_{\mathrm{MMD}}^2$ loss between model and target expectation values computed from the training set. After classical training, the learned GBBM is deployed for \emph{quantum sampling}, where measurements on the quantum device produce generated samples.}
        \label{fig:gbbm_schematic}
    \end{figure}

    Beyond trainability, photonic architectures are also attractive from a hardware perspective. Photonic experiments can operate at repetition rates in the range 1-100 MHz, as set by optical source and detector technologies~\cite{cogan2021deterministicsourceindistinguishablephotons,liu2025robustquantumcomputationaladvantage}. In contrast, while qubit-based platforms support gate operations on nanosecond timescales, the effective inference rate for complete circuit execution (including measurement and reset) is typically limited to the kilohertz regime on current devices~\cite{Arute_2019,rajeev_2022,Stano_2022}. This separation in achievable repetition rates suggests that photonic platforms may offer substantially higher inference throughput in near-term implementations. In addition, Gaussian Boson Sampling architectures appear particularly promising from a scalability perspective, as illustrated by recent large-scale demonstrations on the Jiuzhang 4.0 processor using 1024 input squeezed states and 8176 output modes~\cite{liu2025robustquantumcomputationaladvantage}.

    The paper is organized as follows. In \Cref{sec:framework}, we review the preliminaries, introduce the GBBM model and derive the parity operator and threshold operator formulation of the $\mathrm{MMD}^2$ training objective. Subsequently, \Cref{sec:classical_training_algorithm} details the classical training algorithm for the GBBM based on parity operators.
    In \Cref{sec:benchmarks}, we present a numerical benchmark demonstrating scalability and practical performance with up to $805$ bosonic modes, and more than $10^6$ trainable parameters. Finally, we conclude in \cref{sec:conclusion} with a discussion of limitations, experimental considerations, and possible extensions beyond the Gaussian setting.

\section{Framework of Gaussian Bosonic Born Machines}\label{sec:framework}
    This section presents a framework for classically trainable Born machines in the bosonic setting by employing Gaussian states.
    First, we provide the necessary background on bosonic Gaussian states, Gaussian Boson Sampling and classical training strategies adapted to a photonic domain. We then present a reformulation of the \mmd\ loss function using operator expectation values relevant in bosonic systems.
    Finally we introduce our generative model construction, which is designed to be trained efficiently using classical computing resources and deployed on quantum hardware.

    \subsection{Gaussian states and Gaussian Boson Sampling}\label{ssec:gaussian_states_and_gaussian_boson_sampling}
        Here we give a basic introduction of relevant concepts needed for introducing our proposed quantum generative model. For a more detailed background on Gaussian states we refer to Refs.~\cite{serafini,adesso:2014} and Appendix~\ref{app:gaussian}; concerning Gaussian Boson Sampling, we refer to Refs.~\cite{hamilton2017,Kruse_2019}.

        Gaussian states represent a fundamentally important and experimentally accessible class of continuous-variable quantum states.
        In particular, they can be fully characterized by their displacement vector $\bm{\mu}$ and covariance matrix $\Sigma$ defined elementwise by
        \begin{align}\label{eq:meancov}
        \begin{split}
            \mu_i &\coloneqq \Tr \left [\, \rho \, \hat{\chi}_i  \right ], \\
            \Sigma_{ij} &\coloneqq  \Tr \left [ \, \rho \, \{ \hat{\chi}_i - \mu_i \mathbbm{1}, \hat{\chi}_j - \mu_j \mathbbm{1} \} \right ],
        \end{split}
        \end{align}
        where $\{A, B\} \coloneqq AB + BA$ is the anticommutator and $\bm{\chi}$ is a vector of quadrature operators defined as $ \bm{\chi} = [\hat{\chi}_1, \dots \hat{\chi}_{2d}]^T = [\hat{x}_1, \dots \hat{x}_d, \hat{p}_1, \dots, \hat{p}_d]^T$.
        We call Gaussian states with displacement vector $\bm{\mu} = \bm{0}$ \emph{centralized} (or non-displaced) Gaussian states.
        
        First introduced in Refs.~\cite{hamilton2017, Kruse_2019}, the Gaussian Boson Sampling protocol demonstrates quantum advantage by sampling the photon-number distribution of a randomly generated Gaussian state.
        Similarly to standard Boson Sampling with Fock-state inputs~\cite{aaronson2011bosonsampling}, generating samples from these photonic output distributions is widely believed to be classically hard in the relevant parameter regimes.
        The complexity of taking a sample from a centralized Gaussian state is characterized by the evaluation of a matrix function called the hafnian~\cite{Caianiello1973}, which, in general, is a difficult task using classical resources. Hence, it is widely believed that performing an ideal Gaussian Boson Sampling experiment is intractable on a classical computer~\cite{Grier_2022,li2025complexity_transition_dgbs}.
        
        Traditionally, Gaussian Boson Sampling experiments use photonic circuits composed of squeezing gates and multimodeinterferometers~\cite{photonicadvantage1,zhong2021phase,borealis,liu2025robustquantumcomputationaladvantage}, resulting in pure centralized Gaussian states in ideal conditions.
        However, as we explain in the subsequent sections, our circuit ansatz also uses displacement gates, resulting in displaced Gaussian states. Consequently, the probability distribution corresponding to the particle-resolved particle number detection is characterized by the loop hafnian~\cite{quesada2021quadratic,Bulmer_2022_threshold}. This is a generalization of the conventional Gaussian Boson Sampling experiment, hence, it is also believed to be classically intractable to simulate in low-displacement regimes~\cite{li2025complexity_transition_dgbs}.

        The current state-of-the-art classical simulation of an ideal Gaussian Boson Sampling experiment can be found in Ref.~\cite{Bulmer_2022}. However, as realistic Gaussian Boson Sampling experiments are currently lossy, one can make the sampling algorithm more efficient by incorporating these losses into the classical sampling algorithm~\cite{changhun_mps,go2025sufficientconditionshardnesslossy}. It should be noted, that local noise in large quantum circuits may have the
        effect of strongly suppressing higher-order correlations, which also leads to an efficient algorithm for reproducing experimental Gaussian Boson Sampling data~\cite{dodd2025fastfrugalgaussianboson}.

    \subsection{Training quantum generative models using MMD$^2$}\label{ssec:classical_training_of_q_gen_models}
        Generative modeling aims to learn a parameterized probability distribution \(q_{\bm{w}}\) that approximates a target distribution \(p\), so that samples drawn from \(q_{\bm{w}}\) are statistically similar to those drawn from $p$. Accordingly, the representation should be paired with a protocol that enables efficient sampling from \(q_{\bm{w}}\).

        At the heart of our approach lies the concept of \textit{Quantum Circuit Born Machines} (QCBMs), where the model distribution is implemented using a parametrized quantum circuit followed by a relevant measurement on all qubits~\cite{benedetti2019generative, PhysRevA.98.062324, coyle2020born}, targeting generative problems over $d$ binary random variables.
        In a typical qubit setting, a QCBM is characterized by a circuit ansatz \(U(\bm{w})\) with trainable parameters \(\bm{w}\). Measurement of all qubits in the computational basis defines the model distribution through the Born rule
        \begin{equation}\label{eq:q_w}
        q_{\bm{w}}(\bm{x}) = \Tr [ U(\bm{w})\ketbra{\bm{0}} U^\dagger(\bm{w}) M_{\bm{x}}],
        \end{equation}
        where $\bm{x}$ is a bitstring with $x_i \in \{0,1\}$ and $M_{\bm{x}} = \ketbra{\bm{x}} = \bigotimes_{i=1}^d \ketbra{x_i}$. The parameters are optimized by minimizing a suitable loss so that the generated samples match the training data.
        Similarly, photonic QCBMs~\cite{Salavrakos_2025} are obtained by initializing the system in a $d$-mode vacuum state instead of $|\bm{0}\rangle$, implementing $U(\bm{w})$ as a photonic circuit ansatz, and defining $M_{\bm{x}}$ as the projector onto measurement outcome patterns labeled by $\bm{x}$. It is worth noting that continuous-variable Born machines can also be realized in photonic systems through homodyne measurements~\cite{cepaite2022,kolarovszki_cvbm}, which provide a continuous-outcome alternative to our discrete measurement schemes.

        Besides the choice of ansatz, another important building block in this framework is the loss function. Explicit loss functions, such as the total variation distance (TVD) or the Kullback–Leibler (KL) divergence, are generally hard to estimate accurately from samples, with sample requirements that often scale poorly with dimension. This can make training prohibitively costly and may further worsen trainability issues such as barren plateaus~\cite{rudolph2024trainability}. A convenient alternative is the \textit{squared maximum mean discrepancy} ($\text{MMD}^2$), which is defined below~\cite{gretton}.
        In this work, we consider binary-valued outcomes $\bm{x}\in\{0,1\}^d$, where $x_i$ denotes the binary outcome of mode $i$.
        Let $p$ and $q_{\bm{w}}$ be two probability distributions over outcomes $\bm{x}\in\{0,1\}^d$, representing the target and model distributions, respectively.
        Given a kernel function $K(\bm{x},\bm{y})$, the $\text{MMD}^2$ between these two probability distributions is defined as
        \begin{equation}\label{eq:mmd2_def}
            \mathrm{MMD}^2(p,q_{\bm{w}})
            \coloneqq
            \mathbb{E}_{\bm{x},\bm{x}'\sim p}\!\left[K(\bm{x},\bm{x}')\right]
            +
            \mathbb{E}_{\bm{y},\bm{y}'\sim q_{\bm{w}}}\!\left[K(\bm{y},\bm{y}')\right]
            -
            2\mathbb{E}_{\bm{x}\sim p,\bm{y}\sim q_{\bm{w}}}\!\left[K(\bm{x},\bm{y})\right],
        \end{equation}
        where $K : \{ 0, 1 \}^d \times \{ 0, 1 \}^d \to \mathbb{R}$ is a kernel function.
        The Gaussian kernel
        \begin{equation} \label{eq:gaussian_kernel_func}
            K_{\sigma}(\bm{x},\bm{y}) \coloneqq 
            \exp\!\left(-\frac{\|\bm{x}-\bm{y}\|_2^2}{2\sigma}\right)             =
            \prod_{i=1}^{d}\exp\!\left(-\frac{(x_i-y_i)^2}{2\sigma}\right)
        \end{equation}
        is a natural choice for the \(\mathrm{MMD}^2\)-based training, since it is characteristic and thus uniquely specifies the underlying probability distribution~\cite{gretton}. 
        
        As shown in Ref.~\cite{rudolph2024trainability}, the $\text{MMD}^2$ loss can also be expressed using expectation values of Pauli-$Z$ strings, which is a key observation leading to classically trainable models.
        However, to date, such reformulations of the loss function $\mathcal{L}_{\rm MMD^2}(p,q_{\bm{w}})$ have been established primarily for qubit-based quantum systems, with a straightforward extension to fermionic models~\cite{bakó2025fermionicbornmachinesclassical}.
        For our purposes, we need to generalize this formulation to the bosonic setting.
        We can start by considering projective measurements with binary outcomes 
        labeled by $\bm{x} \in \{ 0, 1 \}^{d}$, 
        where the $j$th element $x_j$ corresponds to a binary outcome on mode $j$. From such projectors $M_{\bm{x}}$ we can define the observable $O_{A}$ as
        \begin{equation}\label{eq:observable_measurement_relation}
            O_{A} \coloneqq \sum_{\bm{x} \in \{ 0, 1 \}^d } (-1)^{|\bm{x}_A|} M_{\bm{x}} = \bigotimes_{j \in A} (M^{(j)}_0 - M^{(j)}_1) \otimes \mathbbm{1}_{[d] \setminus A}
        \end{equation}
        where $A \subseteq [d]$ is a subset of modes and $\bm{x}_A$ denotes the vector formed from $x$ according to the indices in $A$ and $| \bm{x}_A |$ denotes the Hamming weight of $\bm{x}_A$.
        As detailed in Appendix~\ref{app:mmd_loss_photonic}, using \Cref{eq:observable_measurement_relation}, we can rewrite the MMD$^2$ expression as
        \begin{equation}\label{eq:mmd_parity}
            \mathcal{L}_{\rm{MMD}^2}(p, q_{\bm{w}}) 
            = \mathds{E}_{A \sim p_K} \!\left[(\expval{O_{A}}_p -\expval{O_{A}}_{q_{\bm{w}}})^2\right],
        \end{equation}
        where $\expval{O_{A}}_p$ and $\expval{O_{A}}_{q_{\bm{w}}}$ denote expectation values of the operator string $O_{A}$ under the target and model distributions, respectively.
        In this expression, \(K\) denotes the kernel and \(p_K\) the corresponding distribution induced over operator strings, given by
        \begin{equation}\label{eq:p_K_sigma}
            p_{K_{\sigma}}(A)
            =
            (1-p_{\sigma})^{d-|A|}\,p_{\sigma}^{|A|},
        \end{equation}
        where $p_{\sigma}\coloneqq \frac{1}{2} (1-e^{-1/2\sigma})$.
    
        To generalize this procedure for bosonic systems, we define two natural projective measurements that provide binary outcomes suitable for the MMD$^2$ reformulation. These analogs to the qubit Z-measurement represent common detection schemes in quantum optics:
        \begin{enumerate}
            \item \textbf{Parity detection}~\cite{Gerry_2010}: distinguishes between even and odd particle numbers, defined by the single-mode measurement projectors
            \begin{equation}
                M_0 = \sum_{k=0}^{\infty} \ketbra{2k}, \qquad M_1 = \mathbbm{1} - M_0.
            \end{equation}
            
            \item \textbf{Threshold detection}~\cite{threshold_gbs}: differentiates between the vacuum state (``no-click'') and all excited states (``click''), defined by the single-mode measurement projectors
            \begin{equation}
                M_0 = \ketbra{0}, \qquad M_1 = \mathbbm{1} - M_0.
            \end{equation}
        \end{enumerate}
        Both of these measurement schemes effectively ``discretize'' the infinite-dimensional bosonic Hilbert space into a format compatible with the MMD$^2$-based training of QCBMs.
        Moreover, both of these measurements are feasible to perform in a physical experiment; parity detection is extensively used in the Wigner tomography of quantum states~\cite{hofheinz2009synthesizing,vlastakis2013deterministically},
        while threshold detection can be performed with single-photon detectors that are inexpensive and can be operated at room temperature~\cite{Hadfield2009}.

    \subsection{Gaussian Bosonic Born Machines}\label{ssec:model_construction}
        Inspired by the Gaussian Boson Sampling quantum advantage proposal and recent progress in classically trainable quantum generative models, we propose \textit{Gaussian bosonic Born machines} (\gbbm s), a class of photonic QCBMs formally defined as follows:

        \begin{definition}[Gaussian Bosonic Born Machine]
            A \gbbm\ on $d$ bosonic modes is a quantum generative model with the following components:

            \begin{enumerate}
            \item \textbf{Initial state:} an arbitrary Gaussian state, conveniently the vacuum state $\ket{0}$;
            \item \textbf{Variational circuit:} a parameterized photonic quantum circuit composed of Gaussian operations, including multiport interferometers, squeezing and displacement gates;
            \item \textbf{Measurement:} terminal particle number measurements performed on each mode, which is coarse-grained to yield a discrete bitstring-valued probability distribution.
        \end{enumerate}
        \label{def:model}
        \end{definition}

        Our construction of the variational ciruit ansatz follows the continuous-variable quantum neural network architecture proposed in Ref.~\cite{killoran2019continuous}, but restricts the model to linear transformations by omitting the Kerr nonlinearity. Note that, although multiple layers are formally redundant when only linear transformations are used, their inclusion can still improve the optimization behavior of the model; see the discussion below.
        We write the parameter vector as $\bm{w} \coloneqq (\bm{w}^{(1)}, \dots, \bm{w}^{(L)})$ with $L$ being the number of layers in the ansatz and each layer indexed by $k$ is described by
        \begin{equation}
            \boldsymbol{w}^{(k)} \coloneqq
            \bigl(
            \boldsymbol{\alpha}^{(k)},
            \boldsymbol{\theta}_1^{(k)},
            \boldsymbol{r}^{(k)},
            \boldsymbol{\theta}_2^{(k)}
            \bigr),
        \end{equation}
        where
        \begin{enumerate}
            \item $\bm{\alpha}^{(k)} \in \mathbb{R}^d$ is a displacement in the position direction,
            \item $\bm{\theta}_1^{(k)} \in [0, 2\pi)^{d^2}$ defines the first interferometer by $U_1^{(k)} \coloneqq U(\bm{\theta}_1^{(k)}) \in U(d)$,
            \item $\boldsymbol{r}^{(k)} \in \mathbb{R}^d$ are squeezing parameters of the squeezing gates,
            \item $\bm{\theta}_2^{(k)} \in [0, 2\pi)^{d^2}$ defines the second interferometer by $U_2^{(k)} \coloneqq U(\bm{\theta}_2^{(k)}) \in U(d)$.
        \end{enumerate}
        The interferometers are parameterized using standard linear-optical decompositions. In this work, we employ the Clements decomposition~\cite{clements2017optimal}, although the Reck decomposition~\cite{reck} could be used equivalently.
        Although most Gaussian Boson Sampling experiments focus on centered Gaussian states, incorporating displacement gates into the circuit ansatz removes this restriction and substantially enhances the expressivity of the model. Hence, in total, a single layer of the resulting ansatz containing the Clements interferometer layout has $2d^2 + 2d$ trainable parameters.

        We further note that this variational circuit implements a Bloch-Messiah decomposition~\cite{cariolaro2016bloch} and is therefore universal for preparing arbitrary pure Gaussian states with a single layer from the vacuum state. 
        However, the brickwork structure of this ansatz limits information spreading during training. To overcome this obstacle, we first adapt the method discussed in Ref.~\cite{bakó2025fermionicbornmachinesclassical} for a fermionic case and implement multiple layers with independent trainable parameters. This technique was found to enhance the optimization properties of the model associated with overparameterization, while still retaining an efficient classical description of the Gaussian state via its mean vector and covariance matrix. 
        To extend these results, we also introduce an alternative ansatz construction that, instead of a brickwork architecture, implements two-mode gates via beamsplitters along a pre-defined arbitrary graph (with edge ordering specified due to the non-commuting nature of the beamsplitters). As we show in our numerical demonstrations, this can be used to improve information spreading while providing additional flexibility in choosing the number of trainable parameters.
        
        The expectation values of the required operator strings can be computed efficiently on a classical computer; the specific details of this procedure are presented in the next section. Once classical training is complete, the resulting circuit ansatz provides a full description of a pure Gaussian state. In particular, this state can be compiled on a photonic processor using only $\mathcal{O}(d^2)$ linear optical gates in the worst case, a complexity that remains independent of the number of layers used during the classical optimization phase. 
        For inference, the optimized parameters are deployed onto a photonic quantum device equipped with photon number detectors. To generate the discrete data, we only use a coarse-graining of the raw photon counts, possibly circumventing the experimental requirement for photon-number-resolving detectors. Specifically, we consider two types of detectors: parity detectors, governed by the transformation $n\mapsto\frac{1}{2}(1-(-1)^n)$ where $n$ is the raw photon count; and threshold detectors, defined by the mapping
        $$n \mapsto \begin{cases}
        0 & \quad n = 0,\\
        1 & \quad n \geq 1.
        \end{cases}$$
        Either choice of detector effectively transforms the underlying Gaussian state into a discrete bitstring distribution suitable for generative modeling tasks.

        The threshold variant of Gaussian Boson Sampling is widely believed to be computationally intractable to simulate classically~\cite{threshold_gbs}. Because parity and threshold measurements are equivalent in the dilute regime---where the number of modes scales at least quadratically with the total number of photons---this computational hardness naturally extends to the parity-based variant. Note, however, is that sufficiently large displacements can induce a complexity transition, rendering the sampling task classically efficient~\cite{li2025complexity_transition_dgbs}. Moreover, the initial training parameters are typically chosen so that the average photon number of the resulting state scales linearly with the number of modes. Assuming, in addition, that the training data are such that the average Hamming weight scales linearly with the bitstring length, we expect quantum inference to typically operate in the linear regime (See Ref.~\cite{bouland2025complexitytheoreticfoundationsbosonsamplinglinear} for proofs of hardness of Boson Sampling and Bipartite Gaussian Boson Sampling with particle number detectors).

\section{Efficient classical training of parity GBBMs}\label{sec:classical_training_algorithm}
    We continue by presenting an efficient classical training algorithm for GBBMs based on parity operator strings. The key observation is that expectation values of parity strings on Gaussian states can be computed efficiently on a classical computer, independent of string length. Consequently, GBBMs admit an efficient classical training method.

    Our approach starts by tracking the evolution of the mean vector $\bm{\mu}$ and the covariance matrix $\Sigma$ of the underlying Gaussian quantum state. The input state is the vacuum state, characterized by
    \begin{equation}
        \bm{\mu} = \bm{0}, \qquad \Sigma = \mathbbm{1}_{2d},
    \end{equation}
    and we apply a sequence of $L$ parameterized Gaussian transformations corresponding to the layers of the photonic circuit.
    The Gaussian state parameters transform under each layer as
    \begin{equation}\label{eq:gaussian_transformation}
        \bm{\mu} \;\mapsto\; S \bm{\mu} + \bm{r}, 
        \qquad 
        \Sigma \;\mapsto\; S \, \Sigma \, S^{T},
    \end{equation}
    where $\bm{r}$ is the displacement vector and $S$ is the real symplectic matrix associated with the corresponding Gaussian operation (composed of interferometers and squeezing). Since each update in \Cref{eq:gaussian_transformation} involves matrix multiplications of size $2d\times 2d$, the full evolution of the state can be simulated classically in $\mathcal{O}(L d^{3})$ time, where $L$ is the number of layers.

    After computing the final Gaussian state, we evaluate the parity-operator expectation values required for the loss function defined in \cref{eq:mmd_parity}. During training, a fixed number of parity-operator strings are sampled from the distribution specified by the Gaussian kernel from \cref{eq:p_K_sigma}. As discussed in Appendix~\ref{app:parity_operator_exp_values}, the expectation value of a parity-string operator $\Pi_{A}$, acting on a subset of modes $A \subseteq [d]$, admits a closed-form expression:
    \begin{equation}\label{eq:parity_exp_value}
        \Tr\!\left[ \Pi_{A}\, \rho \right]
        =
        \frac{
        \exp\!\left(- \bm{\mu}_{A}^{T} \Sigma_{A}^{-1} \bm{\mu}_{A} \right)
        }{
        \sqrt{\det \Sigma_{A}}
        },
    \end{equation}
    where $\bm{\mu}_{A}$ and $\Sigma_{A}$ denote the mean vector and covariance matrix of the reduced Gaussian state $\Tr_{[d]\setminus A}[ \rho]$, obtained by restricting $\bm{\mu}$ and $\Sigma$ to the modes indexed by $A$.
    Crucially, this formula allows parity-string expectation values to be computed efficiently using only elementary numerical methods. For a parity string of length $\ell = |A|$, the computational cost of evaluating the corresponding expectation value scales as $\mathcal{O}(\ell^{3})$, dominated by the inversion (or Cholesky factorization) of the reduced covariance matrix $\Sigma_{A}$. In practice, however, $\ell \ll d$ for the parity strings of interest, so this step is typically inexpensive. Instead, the overall computational cost might often be dominated by the construction of the full covariance matrix of the $d$-mode Gaussian state. As a result, expectation value estimation remains scalable even for photonic systems with a large number of modes, with the bottleneck shifting from the parity string expectation value evaluation to the covariance matrix construction. Using the backpropagation algorithm, these observations allow for efficient gradient computation. This is summarized in the following theorem.
    \begin{theorem}[Classical trainability of parity GBBMs]\label{thm:classical_trainability_parity}
        For a \gbbm\ over $d$ bosonic modes, as defined in \cref{def:model}, the expectation value of any length-$\ell$ parity string $\Tr\!\left[\Pi_{A}\,\rho\right]$ can be evaluated classically in $\mathcal{O}(\ell^{3})$ time, given the corresponding mean vector and covariance matrix characterizing the final Gaussian state. Moreover, these objects can be obtained in $\mathcal{O}(L d^{3})$ time using $\mathcal{O}(d^{2})$ memory, where $L$ is the number of layers. Consequently, the model admits efficient classical training via automatic differentiation and backpropagation using an approximation of the \mmd\ loss function.
    \end{theorem}
    \begin{proof}
        Evaluating \Cref{eq:parity_exp_value} takes $\mathcal{O}(\ell^3)$ time, as the computation time is dominated by the determinant in the denominator.
        Storing all the data characterizing a Gaussian state requires $\mathcal{O}(d^2)$ memory, which is also sufficient for computing the final covariance matrix from the initial covariance matrix using in-place updates. Updating the covariance matrix requires $\mathcal{O}(L d^3)$ steps, because applying an interferometer to a Gaussian state amounts to an update of the covariance matrix requiring $\mathcal{O}(d^3)$ time.
    \end{proof}

    \noindent For completeness, we also provide a rough sketch of the proposed algorithm for estimating $\text{MMD}^2$ in \Cref{alg:parity_mmd_gbbm}. 
    The algorithm works by computing the mean vector $\bm{\mu}$ and the covariance matrix $\Sigma$ of the Gaussian state from the ansatz parameters $\bm{w}$. From the Gaussian state representation we can compute the expectation value of each parity string, while the target expectation values are obtained from the bitstrings sampled from the target distribution $p$. The average squared difference between the model and target expectation values yields the final \mmd\ estimate.

    \begin{algorithm}[H]
    \caption{Classical evaluation of the MMD$^2$ for a parity GBBM}
    \label{alg:parity_mmd_gbbm}
    \begin{algorithmic}[1]
    \Require
    Parameters $\bm{w} = (\bm{w}^{(1)}, \dots, \bm{w}^{(L)})$, with
    $\bm{w}^{(k)} = (\bm{\alpha}^{(k)}, \bm{\theta}_1^{(k)}, \bm{r}^{(k)}, \bm{\theta}_2^{(k)})$;
    bitstrings $\{\bm{x}^{(n)}\}_{n=1}^{N}$ sampled from the target distribution $p$;
    Gaussian kernel bandwidth $\sigma$;
    number of parity strings $M$.
    \Ensure Estimate of $\mathcal{L}_{\mathrm{MMD}}(p, q_{\bm{w}})$.
    
    \State $\bm{\mu}\gets \bm{0}$,\;\; $\Sigma\gets \mathbbm{1}_{2d}$ \Comment{Vacuum state initialization}
    \For{$k=1,\dots,L$} \Comment{Compute Gaussian state using \Cref{eq:gaussian_transformation}}
        \State $U_1^{(k)}\gets U(\bm{\theta}_1^{(k)})$, \;\; $U_2^{(k)}\gets U(\bm{\theta}_2^{(k)})$
        \State Construct the symplectic matrix $S^{(k)}$ from $(U_1^{(k)},\bm{r}^{(k)},U_2^{(k)})$
        \State Construct the displacement vector $\bm{t}^{(k)}$ from $\bm{\alpha}^{(k)}$
        \State $\bm{\mu}\gets S^{(k)}\bm{\mu}+\bm{t}^{(k)}$ \Comment{Evolve displacement vector}
        \State $\Sigma\gets S^{(k)}\, \Sigma \, (S^{(k)})^{T}$ \Comment{Evolve covariance matrix}
    \EndFor

    \State $S \gets 0$

    \For{$m=1,\dots,M$} \Comment{Compute parity string expectation values}
        \State Sample $A$ from $p_{K_\sigma}$ corresponding to a parity string
        \State $\expval{\Pi_{A}}_p \gets \frac{1}{N}\sum_{n=1}^{N}(-1)^{|\bm{x}_A^{(n)}|}$ \Comment{Compute target parity}
        \State Obtain $(\bm{\mu}_{A},\Sigma_{A})$ by restricting $(\bm{\mu},\Sigma)$ to modes in $A$
        \State $\expval{\Pi_{A}}_{q_{\bm{w}}} \gets
        \exp\!\left(-\bm{\mu}_{A}^{T}\Sigma_{A}^{-1}\bm{\mu}_{A}\right)\big/\sqrt{\det(\Sigma_{A})}$ \Comment{Compute model parity using \Cref{eq:parity_exp_value}}
        \State $S \gets S + \left[\expval{\Pi_{A}}_p - \expval{\Pi_{A}}_{q_{\bm{w}}}\right]^{2} $
    \EndFor
    
    \State \Return $S / M$ \Comment{Return the $\text{MMD}^2$ estimate}
    \end{algorithmic}
    \end{algorithm}

\section{Numerical experiments} \label{sec:benchmarks}
    In this section, we demonstrate the scalability of our generative model and training framework, benchmarking it against well-established classical methods. We first introduce all models and their tunable hyperparameters, followed by describing the benchmark datasets together with the corresponding training and test results. All \gbbm\ trainings rely on our Python implementation~\cite{gaussian_bosonic_born_machines_repository} written using JAX~\cite{jax2018github} for automatic differentiation. These demonstrations aim to study certain important aspects of our framework as well as to compare the performance to classical models. In our assessment of a trained models performance, we rely on two metrics: the covariance matrix of the corresponding probability distributions for visual investigations, and the \mmd\ evaluated with respect to a test set using different $\sigma$ bandwidths. The covariance matrix of a probability distribution $p$ (not to be confused with the covariance matrix of a Gaussian state) is defined elementwise as
    \begin{equation}
        [\operatorname{cov}_{\p}]_{ij} \coloneqq \mathbb{E}_{\bm{x} \sim p} [x_i x_j] - \mathbb{E}_{\bm{x} \sim p} [x_i] \; \mathbb{E}_{\bm{x} \sim p} [x_j].
    \end{equation}
    Since we do not have access to the true underlying probability distribution of the target datasets, but only to finite set of training and test samples, the visualized covariance matrices for the data are necessarily empirical estimates. In contrast, for a trained GBBM parameterized by a Gaussian state $\rho$, the exact analytical covariance matrix of the generated distribution can be computed directly from the expectation values of local operators.

    \subsection{Benchmark models}    \subsubsection{\gbbm\ models}
    We use a varying number of layers and interferometer layouts for the GBBM model as in \cref{def:model}, where the beamsplitter and phaseshifter angle parameters are initialized uniformly at random from $[0, 2\pi)$, and the displacement and squeezing amplitudes are drawn from a normal distribution $\mathcal{N}(0, 0.1)$. Besides choosing the right $\sigma$ bandwidths for the Gaussian kernel function from \Cref{eq:gaussian_kernel_func} (using the median heuristic from Ref.~\cite{garreau2017large}), here we also tune the learning rate and the number of training iterations. While compatible with both gradient-based and gradient-free optimizers, in our implementation, we tune the parameters using the Adam optimizer~\cite{adam2014method}. Our model relies primarily on terminal parity measurements, but we also showcase an alternative strategy in \cref{sec:threshold_experiment} based on threshold detection.
    
    \subsubsection{Restricted Boltzmann machine}
    As a classical baseline learning model, we consider the \textit{Restricted Boltzmann Machine} (RBM) architecture~\cite{ackley1985learning, murphy2012machine}. RBMs form a class of energy-based models defined on a bipartite graph comprising $n$ visible and $m$ latent binary variables. Given a weight matrix $W$ and bias vectors $\bm{a}, \bm{b}$, the energy of the system is given by
    \begin{equation}
    E(\bm{v},\bm{h}) = -\bm{a}^{T}\bm{v} -\bm{b}^{T}\bm{h} - \bm{v}W\bm{h},
    \end{equation}
    where $\bm{v}$ and $\bm{h}$ denote the state of the visible and hidden units. This energy function determines the joint probability distribution according to the Boltzmann distribution
    \begin{equation}
        q(\bm{v},\bm{h}) =  \frac{1}{Z} e^{-E(\bm{v},\bm{h})},
    \end{equation}
    where $Z$ is the normalization constant. The probability of observing a specific visible state $\bm{v}$ is obtained by marginalizing over the latent space:
    \begin{equation}
    q(\bm{v}) = \sum_{\bm{h}} q(\bm{v},\bm{h}).
    \end{equation}
    Our training procedure adopts the strategy from Ref.~\cite{recio2025train}, using persistent contrastive divergence via scikit-learn's \texttt{BernoulliRBM} class~\cite{scikit-learn, qml_benchmarks}. As a key hyperparameter, we scan the number of hidden units $m \in [\lfloor n/2 \rfloor, n]$ and choose the architecture that minimizes the \mmd\ metric accross various bandwidths $\sigma$. 

    \subsubsection{Empirical Chow-Liu tree approximation}
    To capture the basic statistics of the dataset, we construct a tree-structured Bayesian network using the Chow-Liu algorithm~\cite{chow1968approximating} via the \texttt{pgmpy} library~\cite{Ankan2024}. This method approximates the joint distribution by maximizing mutual information to form a maximum weight spanning tree, rooting it at the first variable, and fitting parameters through maximum likelihood estimation. The resulting tree topology allows for efficient and exact sampling because each node has a single parent. We evaluate the performance of the model by measuring the \mmd\ distance between a large, fixed-size set of samples and the original test set. Although this provides a computationally efficient baseline, its accuracy is limited to smaller problems or datasets that align with a tree-like correlation structure.

    \subsection{Exploring overparametrization and alternative ans\"atze}
    \begin{figure}[t]
        \centering
        \includegraphics[width=0.9\linewidth]{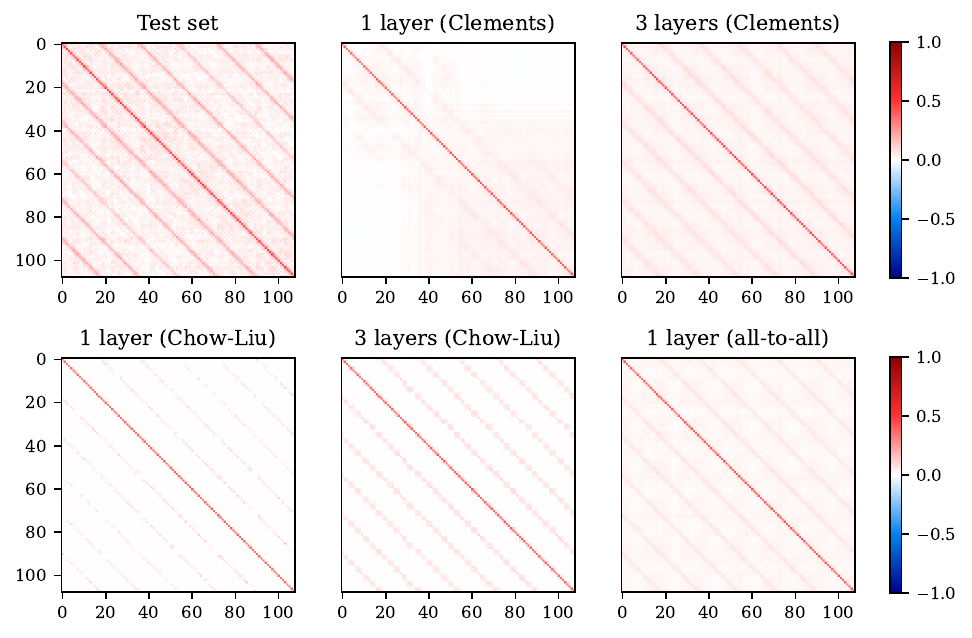}
        \caption{\textbf{The effect of ansatz choice on the performance \gbbm s.} The models were trained on the equilibrium states of a cellular automaton over a $6\times 12$ grid. \gbbm s were constructed with an increasing number of layers and different interferometer layouts (Clements, Chow-Liu and all-to-all layouts). }
        \label{fig:ca_experiment}
    \end{figure}

    Similarly to the method described in Ref.~\cite{bakó2025fermionicbornmachinesclassical}, we start by investigating the effect of including multiple layers in the \gbbm\ model, but also extend it to alternative interferometer layouts. 
    For this, we consider a dataset constructed from the equilibrium states of a \textit{cellular automaton} (CA) on a $6 \times 18$ grid. The samples were obtained by starting the CA in different uniformly random states and evolving them for $1000$ steps under Conway's Game of Life rules~\cite{adamatzky2010game}, then storing the final state. This process was repeated several times to collect a training and test set of $1000$ samples each, after disregarding the all-$0$ states. We present a few examples from the training dataset in Appendix~\ref{appendix:training_examples}.

    The GBBMs were trained for $5000$ episodes with an initial learning rate of $0.001$, using two kernel bandwidths $\sigma \in \{ 8.72, 17.44 \}$. We start by studying the information spreading in the \gbbm\ ansatz and analyze the covariance matrix of the trained models. We consider $3$ distinct interferometer layouts in our experiments. The first is based on the Clements decomposition described in \cref{ssec:model_construction}; a single layer having $23544$ trainable parameters in this case ($108$ modes). The other two layouts are constructed from a pre-specified graph by applying beamsplitters between mode pairs corresponding to the ordered set of edges, inspired by the probabilistic graphical model approach of \cite{bako2024problem}.
    We consider two extreme cases: a tree graph obtained using the Chow-Liu approximation of the dataset and a complete graph, where the edge-ordering corresponds to a uniformly random permutation. A single layer in the ansatz constructed using these graphs admits $860$ and $11664$ parameters, respectively.
    
    The covariance matrices of the trained models are shown in \cref{fig:ca_experiment}, where Chow-Liu and all-to-all refer to the tree and complete graph interferometer layouts. These results signify the importance of the choice of parametrization of Gaussian states in the context of generative learning. 
    The observation that a single layer obtained using the Clements layout is close to diagonal and that this can be improved by including multiple layers is consistent with the fermionic case, noted by the authors in Ref.~\cite{bakó2025fermionicbornmachinesclassical}.
    In the current work, however, we extend these results by showing that alternative parametrizations (with significantly less parameters) can be used to circumvent this unwanted concentration and improve information spreading during training. Even a single layer with the Chow-Liu interferometer layout lacks the covariance concentration behaviour exhibited by the Clements layout (while potentially introducing other biases), despite a $27$-fold reduction of the number of trainable parameters.

    \subsection{USPS handwritten digit experiment}
    For our second experiment, we consider $16\times16$ grayscale images from the USPS dataset of handwritten digits~\cite{hull1994database}. The images are flattened and binarized to obtain bitstrings of length $256$ and divided into a training set of size $19999$ and a test set of $1500$ samples. A few examples of the binarized training set are shown in Appendix.~\ref{appendix:training_examples}.

    The fine-tuned \gbbm\ has $3$ layers of the all-to-all interferometer layout as in the previous section, having $394752$ trainable parameters in total. The model was trained with a learning rate of $0.001$ for $5000$ episodes. The \mmd\ loss was estimated for three bandwidths $\sigma \in \{6.16, 12.33, 24.66\}$ using $10^4$ parity strings each. For comparison, we trained an RBM with $188$ hidden units, where this number was chosen as the best performing in the tuning range. The test results, obtained using $5$ independent estimates for each $\sigma$ bandwidths, are shown in \cref{fig:image_experiment}. Our \gbbm\ model achieves good results, outperforming the fine-tuned RBM model. The Chow-Liu approximation completely fails, showing worse \mmd\ values than the uniformly random samples, as the dataset does not exhibit a strong tree correlation structure.

    \begin{figure}[t]
    \centering
    \includegraphics[width=0.9\linewidth]{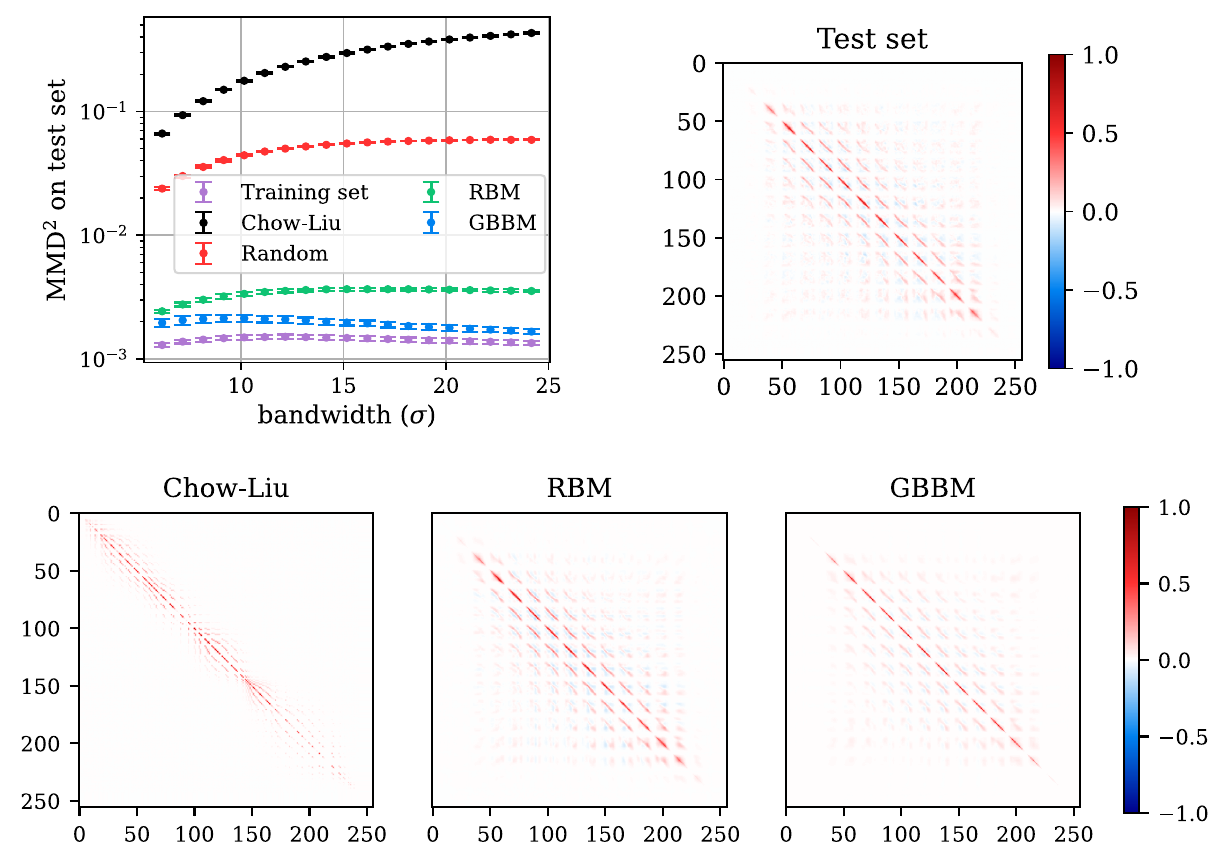}
    \caption{\textbf{Test results on the USPS dataset.} The dataset containing $16 \times 16$ grayscale images was converted into $256$-long bitstrings. A \gbbm\ was trained with $3$ layers of the all-to-all interferometer layout and the fine-tuned RBM was trained with $188$ hidden units for comparison. Markers denote mean values from $5$ independent estimates and errorbars denote the corresponding standard deviations.}
    \label{fig:image_experiment}
\end{figure}

    \subsection{Genomic experiment}

    \begin{figure}[t]
        \centering
        \includegraphics[width=.9\linewidth]{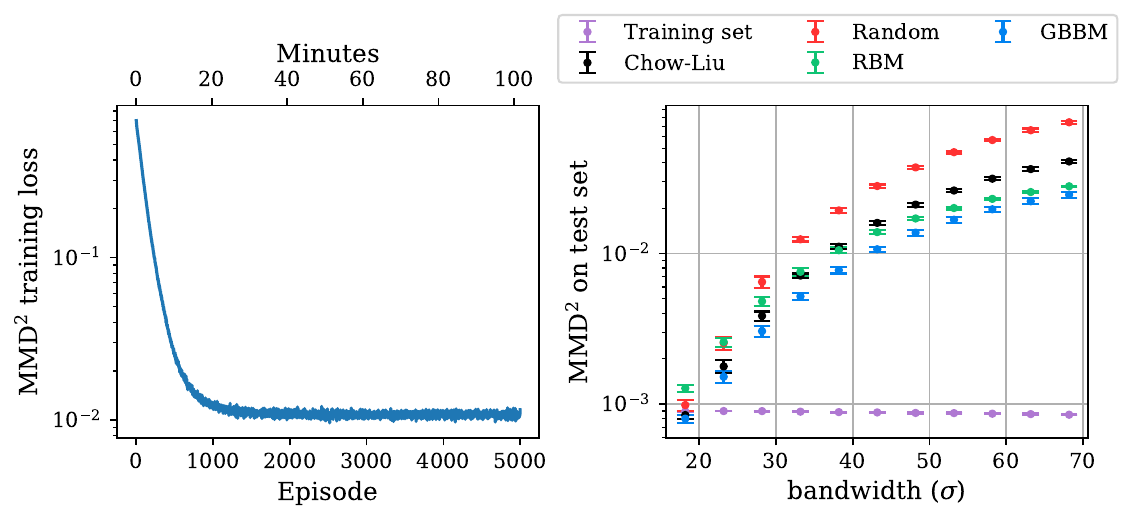}
        \caption{\textbf{Test results on the genomic dataset.} The models were trained on $805$-bit long bitstrings. GBBMs have a single layer with $1297660$ trainable parameters, and we trained an RBM with $492$ hidden units for comparison. Markers denote mean values from $5$ independent estimates and errorbars denote the corresponding standard deviations.}
        \label{fig:genomic_experiment}
    \end{figure}

    For a large scale benchmark, we consider a genomic dataset containing bitstrings of length $805$, where each bit in the sample indicates the presence or absence of the variant allele in the given individual~\cite{recio2025train, bowles2025genomic, yelmen2021creating}. This dataset contains $5008$ samples collected from $2004$ individuals from the Human Genome Project and is divided into training and test datasets with a ratio of $1\!:\!3$.

    We first train a GBBM with $805$ photonic modes using a single layer with the Clements interferometer layout, having $1{,}297{,}660$ trainable parameters. The learning rate is set to $10^{-3}$, and the model is trained for $5000$ episodes, although convergence is reached in less than $2000$ steps. In each training step, we sample $3\times 10^4$ parity strings from \Cref{eq:p_K_sigma} employing $3$ different kernel bandwidths $\sigma \in \{18.11, 36.22, 72.44\}$ obtained using the median heuristic~\cite{garreau2017large} and use the corresponding expectation values to compute the \mmd\ estimate. Notably, the entire training process completed in approximately $100$ minutes using a single NVIDIA A100 GPU. The training curve in both the training episodes and time is shown in \cref{fig:genomic_experiment} (left).

    For comparison, we also train a RBM with $492$ hidden units selected via hyperparameter tuning. The RBM is trained for $1000$ episodes with a learning rate of $10^{-2}$.
    In addition, we also consider the Chow-Liu approximation of the dataset and uniformly random samples as comparisons, computing \mmd\ to the test set based on $10^5$ samples. Finally, we also compute the \mmd\ between the training set and the test set, serving as the ground truth reference.

    After training, we evaluate the \mmd\ of each model with respect to the test set for a range of $\sigma$ kernel bandwidths and display the results in \cref{fig:genomic_experiment} (right). We emphasize that all these discrepancies were evaluated using the \mmd\ formulated as expectation values from \Cref{eq:mmd_parity}. Notably, our GBBM outperforms statistical approximations (Chow-Liu tree and uniformly random samples) and the RBM in the relevant bandwidth regime using a comparable hyperparameter tuning.

\subsection{Threshold GBBMs can learn near-balanced datasets} \label{sec:threshold_experiment}

    There is a notable limitation of parity GBBMs arising from the nature of the parity operator. Since expectation values of parity-operator strings correspond to the central values of reduced Wigner functions, they can only take positive values. Negative values can be obtained by flipping the dataset; therefore, when the dataset is dominated by low or high Hamming weight samples, the model is expected to perform well, as demonstrated previously. However, for balanced datasets containing comparable numbers of $0$s and $1$s, the parity-based model is not sufficient on its own, as the expectation values cannot attain zero, and approximating zero requires higher energy. Although this limitation can be partially mitigated by classical postprocessing, we instead consider an alternative coarse-graining of particle-number measurement outcomes based on threshold detection~\cite{Bulmer_2022_threshold}.
    This threshold detection-based variant is expected to provide significantly improved performance on near-balanced datasets. The classical algorithm for estimating threshold-string expectation values is detailed in Appendix~\ref{sec:classical_training_algorithm_threshold}, where we demonstrate that the enhanced expressivity of the model involves a trade-off regarding the scalability of the classical training process.

To demonstrate a generative learning problem, where the parity \gbbm\ fails, we construct a balanced dataset using the Metropolis algorithm on a classical 2D Ising model with periodic boundary conditions. The coupling strength is set to $J=1$ and we used an alternating local field with magnitude $h=0.08$ on a $14\times14$ grid with periodic boundary condition. First, the model was initialised in a random state and evolved using the Metropolis algorithm for $10^6$ steps. After this warmup, a sample was collected by taking a snapshot of the spin configuration after every $2000$ steps until a dataset of $15 \cdot 10^3$ samples were collected which we divided into a training and test set with a $2:1$ split. Crucially, the algorithm requires a temperature which we set to $T=2.4$, close to the critical temperature of the 2D Ising model. As shown in \cref{fig:ising_experiment} (top, left), the resulting dataset is not concentrated only on the low Hamming weight subspace. 

\begin{figure}[t]
    \centering
    \includegraphics[width=0.9\linewidth]{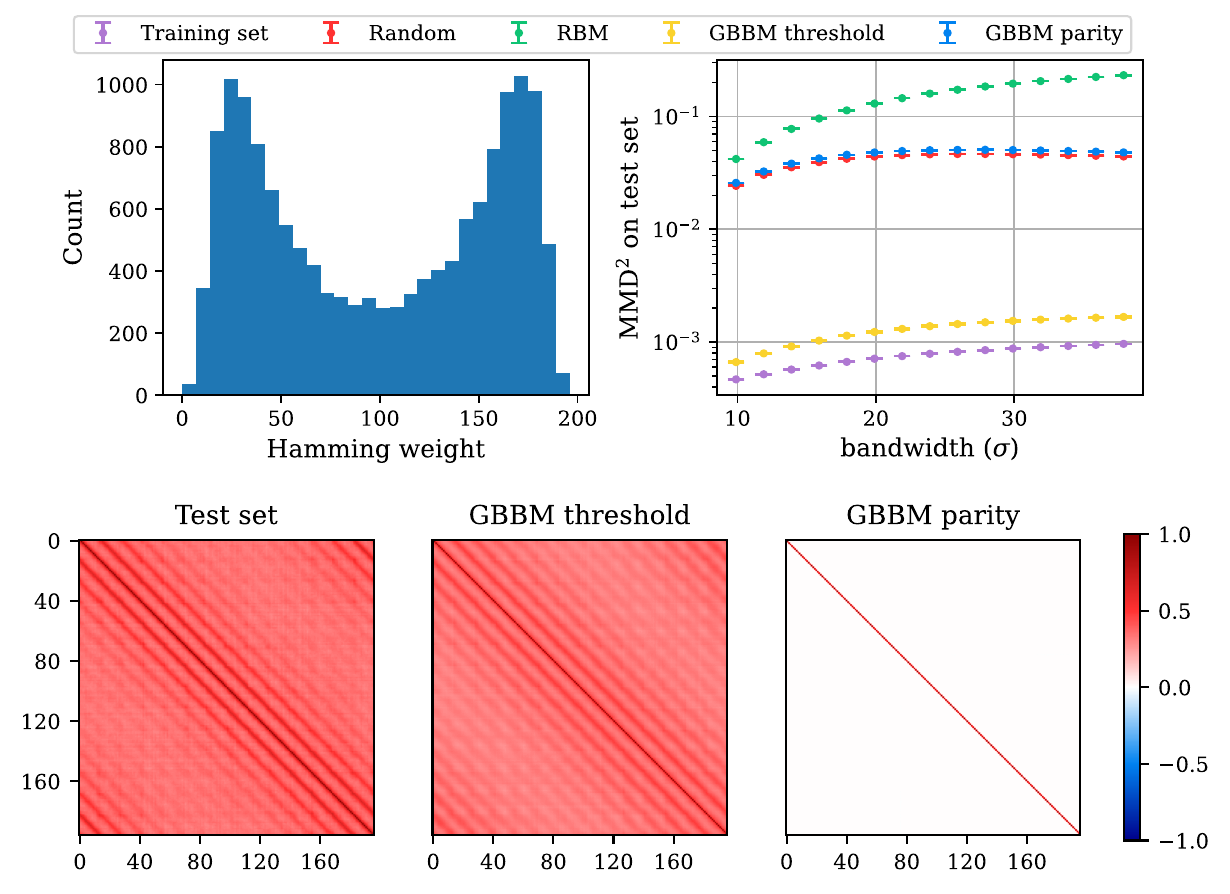}
    \caption{\textbf{Training parity and threshold \gbbm s on the Ising dataset.} The models were trained with at most $7$-long operator-string expectation values. (top, left) Distribution of the Hamming weight of training and test samples. (top, right) Truncated \mmd\ scores of the trained models evaluated up to locality $7$. (bottom) Covariance matrices of the test set and the trained \gbbm s.}
    \label{fig:ising_experiment}
\end{figure}

We train the \gbbm s and the RBM for $5000$ episodes with a learning rate of $0.001$ using $3$ bandwidths $\sigma \in \{10.58, 21.16, 42.33\}$ derived from the median heuristic. However, since the runtime of the classical simulation algorithm depends exponentially on the locality of the threshold operator strings, we impose a cutoff at $\ell =7$. We use this limit for training both types of \gbbm s, and testing all models. Both \gbbm s were constructed with $5$ layers of the Clements ansatz.
The \mmd\ scores and the final covariance matrices of the trained models are shown in \cref{fig:ising_experiment}. The parity \gbbm\ fails to learn this dataset and produces similar results as uniformly random sampling. The tuned RBM performs even worse, but this can be associated with the limitations of contrastive divergence. Notably, this generative problem---close to the critical temperature---is notoriously hard for Boltzmann machines trained with contrastive divergence~\cite{morningstar2018deep}. On the other hand the threshold \gbbm\ shows good results and succesfully learns to capture patterns in the dataset.

\section{Conclusion and outlook}\label{sec:conclusion}

    In this work, we introduced a parity operator-based training framework for bosonic Gaussian quantum models that enables efficient classical evaluation of the $\text{MMD}^2$ loss while retaining a direct connection to experimentally realizable photonic platforms. Our parity-based approach imposes no intrinsic upper bound on the length of parity strings, with the computational cost governed instead by the number of modes and circuit layers. In particular, we demonstrated the practical scalability of the method by successfully training a model on a classical computer with up to $805$ bosonic modes and $1{,}297{,}660$ trainable parameters, using a single GPU in less than $2$ hours. As Gaussian Boson Sampling experiments with thousands of modes are already available on current photonic hardware, our proposed framework is directly relevant to near-term implementations.
    Moreover, in such generative tasks, a key practical advantage of photonic platforms over qubit-based architectures could be their inference throughput: while current qubit-based systems typically operate in the kilohertz regime, photonic implementations are expected to support substantially higher sampling rates~\cite{cogan2021deterministicsourceindistinguishablephotons,liu2025robustquantumcomputationaladvantage}.

    At the same time, realistic experimental conditions impose important constraints. In present-day photonic platforms, optical losses and other noise sources reduce the strength of observable correlations~\cite{dodd2025fastfrugalgaussianboson}, which can limit the effective expressivity of the model. Nevertheless, training can still be performed in the presence of losses and learning can remain meaningful even when experimental imperfections are taken into account. Moreover, Gaussian Boson Sampling is believed to remain classically hard even in the presence of moderate loss~\cite{go2025sufficientconditionshardnesslossy}, suggesting that useful generative models may be achievable without fully ideal hardware. Further reductions in experimental loss would make richer correlations accessible and could therefore improve the practical performance of the model.

    Looking ahead, several directions merit further investigation. While parity and threshold operators constitute a robust and experimentally natural choice of observables, they capture only the coarsest nontrivial features of the resulting photon-number distribution. A natural generalization is to replace the operator strings in our framework with number-resolving observables, leading to a framework based on $n$-ary digits instead of bits. Understanding which classes of such observables admit efficient classical evaluation is an important open question. More broadly, it remains to be explored whether restricted forms of non-Gaussianity can be incorporated into the model, while preserving classical trainability. Additionally, one could shift the underlying quantum advantage scheme entirely; for instance, recent work has explored analogous generative frameworks based on standard Boson Sampling rather than Gaussian Boson Sampling~\cite{kurkin2026, gottlieb2026efficient}.
    In general, identifying the boundary between enhanced expressivity and classical intractability will be crucial for the development of scalable and practically relevant photonic generative models.

\section{Acknowledgement}

    ZK, BB and ZZ would like to thank the support of the Hungarian National Research, Development and Innovation Office (NKFIH) through the KDP-2023 funding scheme, the Quantum Information National Laboratory of Hungary and the grants TKP-2021-NVA-04, TKP2021-NVA-29 and FK 135220. 
    MO acknowledges the support from the European Union’s Horizon Europe research and innovation program under EPIQUE Project GA No. 101135288.
    CO was supported by Global Partnership Program of Leading Universities in Quantum Science and Technology (RS-2025-08542968) through the National Research Foundation of Korea (NRF) funded by the Korean government (Ministry of Science and ICT(MSIT)).
    CO was also supported by the National Research Foundation of Korea Grants (No. RS-2024-00431768 and No. RS-2025-00515456) funded by the Korean government (Ministry of Science and ICT (MSIT)) and the Institute of Information \& Communications Technology Planning \& Evaluation (IITP) Grants funded by the Korea government (MSIT) (No. IITP-2025-RS-2025-02283189 and IITP-2025-RS-2025-02263264).
    ZZ was partially supported by the Horizon Europe programme HORIZON-CL4-2022-QUANTUM-01-SGA via the project 101113946 OpenSuperQPlus100 and the QuantERA II project HQCC-101017733.
    ZZ also acknowledges funding from the Research Council of Finland through the Finnish Quantum Flagship project 358878.
    The C4QEC project is carried out within the IRAP of the Foundation for Polish Science co-financed by the European Union.
    The authors also acknowledge the computational resources provided by the Wigner Scientific Computational Laboratory (WSCLAB).

\newpage

\appendix

\section{Recap: Gaussian states}\label{app:gaussian}

    The current Appendix provides a brief overview of linear optics with an emphasis on Gaussian states and transformations; for comprehensive introductions, see Refs.~\cite{adesso:2014,serafini}.

    A Gaussian state $\rho$ is a photonic quantum state fully characterized by the first and second moments of the quadrature operators $\hat{x}_1,\dots,\hat{x}_d,\hat{p}_1,\dots,\hat{p}_d$.
    These moments are collected into the mean vector $\bm{\mu}\in\mathbb{R}^{2d}$ and the covariance matrix $\Sigma\in\mathbb{R}^{2d\times 2d}$, defined elementwise as
    \begin{subequations}
    \begin{align}
        \mu_{i} &\coloneqq \Tr \left [ \rho \, \hat{\chi}_i  \right ], \\
        \Sigma_{ij} &\coloneqq
        \Tr \left [ \rho \, \{ \hat{\chi}_i - \mu_i \mathbbm{1}, \hat{\chi}_j - \mu_j \mathbbm{1} \} \right ],
    \end{align}
    \end{subequations}
    where $\bm{\chi} = [\hat{\chi}_1, \dots, \hat{\chi}_{2d}]^T$ is the formal vector of quadrature operators defined as 
    \begin{equation}
        \bm{\chi} \coloneqq  [
            \hat{x}_1,
            \dots,
            \hat{x}_d,
            \hat{p}_1,
            \dots,
            \hat{p}_d
        ]^T,
    \end{equation}
    and $d$ is the number of bosonic modes.
    The quadrature operators satisfy the canonical commutation relations 
    \begin{equation}\label{eq:ccr}
        [\hat{\chi}_i, \hat{\chi}_j] = i \Omega_{ij},
    \end{equation}
    where $\Omega$ is the canonical symplectic form
    \begin{equation}
        \Omega \coloneqq \begin{bmatrix} 0 & \mathbbm{1}_{d \times d} \\ -\mathbbm{1}_{d \times d} & 0 \end{bmatrix}.
    \end{equation}
    For $\Sigma$ to describe a physical quantum state, it must be symmetric and satisfy the Robertson-Schrödinger uncertainty relation~\cite{serafini}
    \begin{equation}
        \Sigma + i\Omega \ge 0.
    \end{equation}

    The bridge between the operator formalism and the phase-space representation of quantum optics is provided by the Weyl operator (or displacement operator) $D(\bm{\xi}) = \exp(i \bm{\xi}^T \Omega \bm{\chi})$. The phase-space description of Gaussian states is conveniently expressed in terms of the Weyl operators as follows.
    The Gaussian nature of the quantum state implies that its characteristic function, $\chi_\rho(\bm{\xi}) = \Tr[\rho D(\bm{\xi})]$, takes a quadratic exponential form.
    The Wigner function $W_\rho(\bm{r})$ is the symplectic Fourier transform of the characteristic function, which---up to normalization---represents the distribution of a photonic quantum state in phase space~\cite{wigner_function}. For a Gaussian state characterized by a mean vector $\bm{\mu}$ and a covariance matrix $\Sigma$, the Wigner function is proportional to the displaced multivariate Gaussian distribution~\cite{serafini} and has the form
    \begin{equation}\label{eq:wigner_gaussian}
        W_{\rho}(\bm{r}) =
         \frac{2^d }{\pi^{d}\sqrt{\det \Sigma}}
        e^{- (\bm{r} - \bm{\mu})^T \Sigma^{-1} (\bm{r} - \bm{\mu})},
    \end{equation}
    where $\bm{r} \in \mathbb{R}^{2d}$ represents the phase-space coordinates. Hence, as expected, the mean vector $\bm{\mu}$ represents the displacement of the Gaussian state in the phase space and $\Sigma$ its covariance.
    This representation highlights that Gaussian states are the ``most classical'' of quantum states in the sense that their Wigner functions are non-negative, allowing them to be interpreted as valid probability distributions over phase space.

    We call Gaussian transformations those unitary operations that preserve the Gaussian character of photonic quantum states. Equivalently, in the Heisenberg picture, they act on the vector of canonical ladder (or quadrature) operators as symplectic affine transformations.
    More concretely, the Heisenberg action of the unitary \(U\) is given by
    \begin{equation}
        U^\dagger \bm{\chi} U
        =
        S\,\bm{\chi} + \bm{r}\,\mathbbm{1},
    \end{equation}
    where \(S\) is a real symplectic matrix describing the interferometer and squeezing operations, and \(\bm{r} \in \mathbb{R}^{2d}\) is a displacement vector.
    The symplectic matrix \(S\) satisfies the constraint
    \begin{equation}
        S \Omega S^{T} = \Omega,
    \end{equation}
    ensuring the preservation of the canonical commutation relations from \Cref{eq:ccr}.
    As a direct consequence, Gaussian states are closed under such transformations, and their evolution is completely characterized at the level of first and second moments. Hence, the transformation of a mean vector $\bm{\mu}$ and a covariance matrix $\Sigma$ under a Gaussian transformation \(U\) reads
    \begin{subequations}    
        \begin{align}\label{eq:gaussian_time_evolution}
        \bm{\mu} &\mapsto S\,\bm{\mu} + \bm{r}, \\
        \Sigma &\mapsto S\,\Sigma\,S^{T}.
        \end{align}
    \end{subequations}

\section{Photonic MMD\texorpdfstring{$^2$}{2} loss from operator string expectation values}\label{app:mmd_loss_photonic}
    This Appendix contains a derivation of the expression for the squared maximum mean discrepancy (MMD$^2$) loss in the photonic setting in terms of expectation values of operator strings. The resulting form is the bosonic analog of the qubit Pauli-$Z$ string representation from Ref.~\cite{rudolph2024trainability}, and parallels the fermionic translation from Ref.~\cite{bakó2025fermionicbornmachinesclassical}.
    Although this thought carries out verbatim by replacing the Pauli-$Z$ operator to its bosonic analog in the derivation of Ref.~\cite{rudolph2024trainability} for qubit-based systems, we include a self-contained adaptation for completeness.

    We consider a binary projective measurement on each mode \(j\), with local projectors
    \(
    M_0^{(j)}
    \)
    and
    \(
    M_1^{(j)}
    \),
    and denote the corresponding global projectors by
    \begin{equation}
    M_{\bm{x}} \coloneqq \bigotimes_{j=1}^d M_{x_j}^{(j)},
    \qquad
    \bm{x}\in\{0,1\}^d.
    \end{equation}
    For every subset \(A\subseteq[d]\), we define the associated operator string
    \begin{equation}
        O_A
        \coloneqq
        \sum_{\bm{x}\in\{0,1\}^d} (-1)^{|\bm{x}_A|} M_{\bm{x}}
        =
        \bigotimes_{j\in A}\!\left(M_0^{(j)}-M_1^{(j)}\right)
        \otimes
        \mathbbm{1}_{[d]\setminus A},
    \end{equation}
    where \(\bm{x}_A\) denotes the restriction of \(\bm{x}\) to the indices in \(A\), and \(|\bm{x}_A|\) is its Hamming weight.
    
    Let \(p\) and \(q\) be two probability distributions on \(\{0,1\}^d\). The squared maximum mean discrepancy MMD$^2$ between them is
    \begin{equation} \label{eq:mmd_app}
        \mathrm{MMD}^2(p,q)
        \coloneqq
        \mathbb{E}_{\bm{x},\bm{x}'\sim p}\!\left[K(\bm{x},\bm{x}')\right]
        +
        \mathbb{E}_{\bm{y},\bm{y}'\sim q}\!\left[K(\bm{y},\bm{y}')\right]
        -
        2\mathbb{E}_{\bm{x}\sim p,\bm{y}\sim q}\!\left[K(\bm{x},\bm{y})\right],
    \end{equation}
    For the Gaussian kernel
    \begin{equation}
    K_{\sigma}(\bm{x},\bm{y})
    =
    \exp\!\left(-\frac{\|\bm{x}-\bm{y}\|_2^2}{2\sigma}\right)
    =
    \prod_{i=1}^d
    \exp\!\left(-\frac{(x_i-y_i)^2}{2\sigma}\right),
    \end{equation}
    let us define
    \begin{equation}
    p_{\sigma}\coloneqq \frac12\left(1-e^{-1/(2\sigma)}\right).
    \end{equation}
    Then one readily verifies that
    \begin{equation}\label{eq:singlebit_kernel_walsh}
        \exp\!\left(-\frac{(x_i-y_i)^2}{2\sigma}\right)
        =
        (1-p_{\sigma}) + p_{\sigma}(-1)^{x_i+y_i}.
    \end{equation}
    Substituting this identity into the product form of the kernel yields the subset expansion
    \begin{equation}\label{eq:gaussian_kernel_subset_expansion}
        K_{\sigma}(\bm{x},\bm{y})
        =
        \sum_{A\subseteq[d]}
        p_{K_\sigma}(A)\,
        (-1)^{|\bm{x}_A|}\,
        (-1)^{|\bm{y}_A|},
    \end{equation}
    where
    \begin{equation}
        p_{K_\sigma}(A)
        \coloneqq
        (1-p_{\sigma})^{d-|A|}p_{\sigma}^{|A|}.
    \end{equation}
    Since
    \(
    \sum_{A\subseteq[d]} p_{K_\sigma}(A)=1
    \),
    this defines a probability distribution over subsets of \([d]\).
    It is therefore natural to introduce the following:
    \begin{equation}
        \widehat p(A)
        \coloneqq
        \sum_{\bm{x}\in\{0,1\}^d} p(\bm{x})\,(-1)^{|\bm{x}_A|},
        \qquad
        \widehat q(A)
        \coloneqq
        \sum_{\bm{y}\in\{0,1\}^d} q(\bm{y})\,(-1)^{|\bm{y}_A|}.
    \end{equation}
    Using \Cref{eq:gaussian_kernel_subset_expansion}, we obtain
    \begin{equation}
        \mathbb{E}_{\bm{x}\sim p,\bm{y}\sim q}\!\left[K_{\sigma}(\bm{x},\bm{y})\right]
        =
        \sum_{A\subseteq[d]} p_{K_\sigma}(A)\,\widehat p(A)\,\widehat q(A).
    \end{equation}
    Applying the same identity to all three terms in the definition of \(\mathrm{MMD}^2\), we arrive at
    \begin{equation}
        \mathrm{MMD}^2(p,q)
        =
        \sum_{A\subseteq[d]}
        p_{K_\sigma}(A)\,
        \bigl(\widehat p(A)-\widehat q(A)\bigr)^2
        =
        \mathbb{E}_{A\sim p_{K_\sigma}}
        \!\left[
            \bigl(\widehat p(A)-\widehat q(A)\bigr)^2
        \right].
    \end{equation}
    
    Now let \(\rho_p\) and \(\rho_q\) be quantum states that reproduce the distributions \(p\) and \(q\) using the binary projective measurement \(\{M_{\bm{x}}\}_{\bm{x}\in\{0,1\}^d}\), i.e.,
    \begin{equation}
        \Tr[\rho_p M_{\bm{x}}]=p(\bm{x}),
        \qquad
        \Tr[\rho_q M_{\bm{x}}]=q(\bm{x}).
    \end{equation}
    Then, by the definition of \(O_A\), we can write that
    \begin{equation}
        \Tr[\rho_p O_A]
        =
        \sum_{\bm{x}\in\{0,1\}^d}
        (-1)^{|\bm{x}_A|} p(\bm{x})
        =
        \widehat p(A),
    \end{equation}
    and similarly
    \begin{equation}
        \Tr[\rho_q O_A]
        =
        \sum_{\bm{x}\in\{0,1\}^d}
        (-1)^{|\bm{x}_A|} q(\bm{x})
        =
        \widehat q(A).
    \end{equation}
    Hence, the MMD\(^2\) loss can be written as
    \begin{equation}
        \mathcal{L}_{\mathrm{MMD}^2}(p,q)
        =
        \mathbb{E}_{A\sim p_{K_\sigma}}
        \!\left[
            \bigl(
                \expval{O_A}_{\rho_p}
                -
                \expval{O_A}_{\rho_q}
            \bigr)^2
        \right].
    \end{equation}

\section{Parity operator expectation values of Gaussian states}\label{app:parity_operator_exp_values}
    In this Appendix, we derive the exact formula for the parity operator expectation values for Gaussian states for completeness. The expectation value of the parity operator strings correspond to the central values of the Wigner function~\cite{cahillglauber69}, a well-known fact used extensively in Wigner tomography~\cite{hofheinz2009synthesizing,vlastakis2013deterministically}.

    Let us start by computing the characteristic function of the single-mode parity operator $(-1)^{\hat{n}}$:
    \begin{equation}
        \chi_{(-1)^{\hat{n}}}(x, p) \coloneqq \Tr [ P D(x, p)] = e^{- \frac{x^2 + p^2}{2}}  \sum_{n=0}^\infty (-1)^n L_n(x^2 + p^2) = \frac{1}{2},
    \end{equation}
    where $D(x, p)$ denotes the single-mode displacement operator, $L_n$ denotes the Laguerre polynomial of degree $n$ and we have used the well-known identity
    \begin{equation}
        \sum_{n=0}^\infty (-1)^n L_n(x^2 + p^2) = \frac{1}{2}e^{\frac{x^2 + p^2}{2}}.
    \end{equation}
    Now, consider the global parity operator $\Pi = (-1)^{\hat{n}_{1} + \dots + \hat{n}_{d}}$ defined on a system with $d$ modes. The characteristic function of this operator is simply just
    \begin{equation}
        \chi_{\Pi} (\bm{\xi}) =
        \Tr [\Pi D(\bm{\xi})] = \frac{1}{2^d},
    \end{equation}
    where $D(\bm{\xi})$ denotes the multimode displacement operator.
    Now consider a $d$-mode photonic quantum state $\rho$. Its parity expectation value is
    \begin{align}\begin{split}
        \Tr [ \Pi \rho] 
        &= 
        \frac{1}{\pi^d} \int_{\mathbb{C}^d} \chi_{\rho}(\bm{\xi}) \Tr[\Pi D(-\bm{\xi})] \dd^{2d} \bm{\xi}
        = 
        \frac{1}{(2\pi)^{d}} \int_{\mathbb{C}^d} \chi_{\rho}(\bm{\xi}) \dd^{2d} \bm{\xi}
        =
        \left(\frac{\pi}{2}\right)^d W_{\rho}(0).
    \end{split}\end{align}
    Moreover, if $\rho$ is also a Gaussian state and is parametrized by the mean vector $\bm{\mu}$ and covariance matrix $\Sigma$, we can write
    \begin{equation}
         \Tr [ \Pi  \rho] = 
         \left(\frac{\pi}{2}\right)^d W_{\rho}(\bm{0})
         = \frac{ \exp(- \bm{\mu}^T \Sigma^{-1} \bm{\mu}) }{\sqrt{\det \Sigma}}.
    \end{equation}
    Geometrically, the expectation value of the parity operator is proportional to the value of the Wigner function in the center.
    We note that the expectation value of a parity string corresponding to a subsystem labeled by $A = \{ i_1, \dots, i_\ell \}$ is just the expected parity of the state reduced to that subsystem. More concretely,
    writing
    \begin{equation}
        \Pi_{A} \coloneqq (-1)^{\hat{n}_{i_1} + \dots + \hat{n}_{i_\ell}},
    \end{equation}
    we get that
    \begin{equation}
        \Tr [ \Pi_{A}\, \rho] = \Tr [ \Pi \Tr_{[n]\setminus A} (\rho)] = 
        \frac{ \exp(- \bm{\mu}_{A}^T \Sigma_{A}^{-1} \bm{\mu}_{A}) }{\sqrt{\det \Sigma_{A}}},
    \end{equation}
    where $\bm{\mu}_{A}$ and $\Sigma_{A}$ are the mean vector and covariance matrix corresponding to $\Tr_{[n]\setminus A} [\rho]$, obtained by removing the rows (and columns) corresponding to $[n] \setminus A$.
    
    Of course, calculating the expected parity of an arbitrary Gaussian state is tractable on a classical computer, but we would like to provide the best possible algorithm in order to push the limits of our proposed algorithm.
    In our implementation, the time complexity of calculating the parity of a $d$-mode Gaussian state is $\mathcal{O}(d^3)$, dominated by the Cholesky factorization of the covariance matrix, which involves roughly $d^3/3$ FLOPs. From the Cholesky factor, one can easily obtain both the determinant of $\Sigma$ and the exponent $\bm{\mu}^T \Sigma^{-1} \bm{\mu}$.

\section{Classical simulation of parity Gaussian Boson Sampling}
    In this Appendix, we derive a closed expression for parity-outcome probabilities of general Gaussian states and analyze their classical evaluation. We show how parity-string expectation values can be precomputed and reused in a mode-by-mode sampling scheme, discuss the resulting computational scaling, and identify lossy regimes where parity-based sampling may become competitive with existing classical simulation methods.
    
    Detecting parity $\bm{x}$ in a Gaussian state $\rho$ characterized by a mean vector $\bm{\mu}$ and a covariance matrix $\Sigma$ has probability
    \begin{align}\begin{split} \label{eq:parity_probability}
        p(\bm{x}) &= \Tr \left[
            \rho \bigotimes_{i=1}^d \left(\frac{1 + (-1)^{x_i} (-1)^{\hat{n}_i}}{2}\right)
        \right]
         = \frac{1}{2^d} \sum_{S \subseteq [d]} (-1)^{|\bm{x}_{S}|} \Tr \left[
            \Pi_{S} \rho
         \right]
         \\&= 
         \frac{1}{2^d} \sum_{S \subseteq [d]} (-1)^{ |\bm{x}_{S}|} \frac{ \exp(- \bm{\mu}_{S}^T \Sigma_{S}^{-1} \bm{\mu}_{S}) }{\sqrt{\det \Sigma_{S} }}.
    \end{split}
    \end{align}
    As the same parities can be reused during the mode-by-mode sampling algorithm, it is best to have all the parities calculated in advance for all the $2^d$ possible subsystems, which naively requires at most $\mathcal{O}(d^3 2^d)$ resources. However, due to the data dependency of the Cholesky factorization detailed in Ref.~\cite{kaposi2022polynomialspeeduptorontoniancalculation}, the algorithm does not need to redo the full Cholesky factorization from scratch for every summand. With this strategy, we expect a significant polynomial improvement in the time complexity. Moreover, we can reuse the previously calculated probability during mode-by-mode sampling using
    \begin{equation}
        p(x_n, x_{n-1}, \dots, x_1) = p(x_{n-1}, \dots, x_1) - \frac{1}{2^{n-1}} \sum_{S \subseteq [n-1]} (-1)^{|\bm{x}_{S}|} \Tr \left[
            \Pi_{S \cup \{n\}} \rho
         \right].
    \end{equation}
    However, we note that for simulating ideal Gaussian Boson Sampling, combining coarse-graining with current state-of-the-art classical simulation algorithms for particle-number resolved detection, such as the approach presented in Ref.~\cite{Bulmer_2022} with cost scaling as $\tilde{\mathcal{O}}(2^{d/2})$ in the worst case, are more efficient than using \Cref{eq:parity_probability} directly.

    Notably, the last expression is similar to the loop torontonian $\operatorname{ltor}$ introduced in Ref.~\cite{quesada2021quadratic}. In fact, one can rewrite \Cref{eq:parity_probability} for the special case of $\bm{x} = \bm{1} \coloneqq [1, \dots, 1]^T$ as
    \begin{equation}
        p(\bm{1}) = \frac{(-1)^d \; e^{-\bm{\mu}^T \Sigma^{-1} \bm{\mu}}}{2^d \sqrt{\det \Sigma}} \operatorname{ltor}( \mathbbm{1} -\Sigma^{-1}, \sqrt{2} \Sigma^{-1} \bm{\mu}).
    \end{equation}
    However, as opposed to the threshold detection probabilities, parity-outcome probabilities can be written as an expectation value of parity expectation values that are in the interval $[-1, 1]$. Hence, one can utilize a simple Hoeffding's inequality to get an additive error approximation of this probability, yielding an approximate classical sampling strategy. This can be beneficial in the setting where the Gaussian Boson Sampling experiment is lossy, as higher-order parity-string expectation values are suppressed by losses, an observation already used for devising a classical simulation algorithm for lossy Gaussian Boson Sampling in Ref.~\cite{dodd2025fastfrugalgaussianboson}.

    Moreover, as expected from the discussion in Appendix~\ref{app:mmd_loss_photonic}, we can rewrite the probabilities from \Cref{eq:parity_probability} as the Walsh-Hadamard transform of the parity expectation values as
    \begin{equation}
        p(\bm{x}) \coloneqq \frac{1}{2^d} \sum_{\bm{s} \in \{ 0, 1\}^d } (-1)^{ \bm{x} \cdot \bm{s}} \expval{\Pi_{\xi(\bm{s})}}_{\rho},
    \end{equation}
    where $\xi(\bm{s}) \coloneqq \{ s_i \in [d] : s_i = 1 \}$. Hence, given all the parity expectation values for all $2^d$ subsystems, one can compute the parity detection probabilities in $\mathcal{O}(d 2^d)$ time using the fast Walsh-Hadamard transform~\cite{fast_wh}.

\section{Classical training of threshold GBBMs}\label{sec:classical_training_algorithm_threshold}
    This Appendix outlines the classical simulation algorithm for training threshold GBBMs. Transitioning from parity GBBMs to threshold GBBMs requires replacing parity operator strings with threshold operator strings, a shift that necessitates a fundamentally different set of formulas for the classical training procedure.

    We construct a generic threshold operator string for a subset of modes $A \subseteq [d]$ as follows:
    \begin{equation}
        T_{A} = \bigotimes_{j \in A} (2\ketbra{0_j} - \mathbbm{1}_j) \otimes \mathbbm{1}_{[d]\setminus A}.
    \end{equation}
    To compute the expectation value of $T_A$ by a Gaussian state $\rho$,
    we apply an inclusion-exclusion expansion over the subsets $S \subseteq A$:
    \begin{equation}\label{eq:threshold_operator_expval}
        \expval{T_A}_{\rho} = (-1)^{|A|} \sum_{S \subseteq A} (-2)^{|S|} p_0(S).
    \end{equation}
    In this expression, $p_0(S)$ denotes the probability of detecting the Gaussian state $\rho$ as a vacuum on modes labeled by $S$. This can be computed using the overlap formula~\cite{serafini} as
    \begin{equation}
        p_0(S) = \frac{
            \exp(
                - \frac{1}{2} \bm{\mu}_S^\dagger Q_{S}^{-1} \bm{\mu}_S
            )
        }{\sqrt{\det Q_{S}}},
    \end{equation}
    where $\bm{\mu}_S$ and $Q_S$ denote the mean vector and Husimi covariance matrix of the original Gaussian state restricted to modes defined by $S$. The full Husimi matrix is related to the covariance matrix as $Q = \frac{1}{2}(\Sigma + \mathbbm{1})$.

    With the analytical form of the expectation value established, we can now assess the computational complexity of evaluating these formulas during training. The primary bottleneck lies in the inclusion-exclusion expansion, which requires computing the vacuum detection probabilities for all $2^{|A|}$ subsets. Since evaluating each overlap involves the determinant and inverse of a matrix up to size $|A| \times |A|$, we arrive at the following result regarding the scalability of the threshold GBBMs:

    \begin{theorem}[Classical trainability of threshold GBBMs]
        For a \gbbm\ over $d$ bosonic modes, as defined in \cref{def:model}, the expectation value of any length-$\ell$ threshold string can be evaluated classically in $\mathcal{O}(\ell^{3} 2^\ell)$ time, given the corresponding covariance matrix. This covariance matrix can be computed in $\mathcal{O}(L d^{3})$ time using $\mathcal{O}(d^{2})$ memory. Consequently, the model admits efficient classical training via automatic differentiation and backpropagation, when threshold operator strings  are of length $\mathcal{O}(\log(d))$.   
    \end{theorem}
    \begin{proof}
        The evaluation of the covariance matrix is exactly the same as in \Cref{thm:classical_trainability_parity}.
        Moreover, evaluating \Cref{eq:threshold_operator_expval} requires $\mathcal{O}(\ell^3 2^\ell)$ steps, as the algorithm needs to iterate over $2^\ell$ summands, and each summand requires computing a determinant and an inverse (or solving the equivalent linear system) which requires $\mathcal{O}(\ell^3)$ time to evaluate. Consequently, the model admits efficient classical training for threshold operator strings that are of length $\ell = \mathcal{O}(\log(d))$, as this yields an evaluation time for the threshold operator string expectation values of $\mathcal{O}(\mathrm{poly}(d))$.
    \end{proof}
    
    \noindent While the results in \Cref{sec:threshold_experiment} indicate that the threshold-based approach maintains robustness in regimes where parity-based methods fail to converge, this performance comes at the cost of classical training efficiency. This highlights a fundamental trade-off between the expressivity of the model and its computational overhead during classical training.

    Notably, the parity-based and threshold-based variants converge in the dilute regime, where the number of modes scales at least quadratically with the total photon count ($n = \mathcal{O}(\sqrt{d})$). Given that parity-based GBBMs offer superior classical training efficiency, they can serve as an effective pre-training stage for threshold models. This hierarchical approach significantly reduces the overall computational overhead of training the more expressive threshold-based variant.
\section{Training set visualization}
\label{appendix:training_examples}
Here we present a few examples of the cellular automaton and binarized USPS datasets, naturally defined over $6\times 18$ and $16\times16$ grids, respectively. Random samples are shown in \cref{fig:training_examples}.

\begin{figure}[H]
    \centering
    
    \begin{subfigure}[b]{0.7\textwidth}
        \centering
        \includegraphics[width=\textwidth]{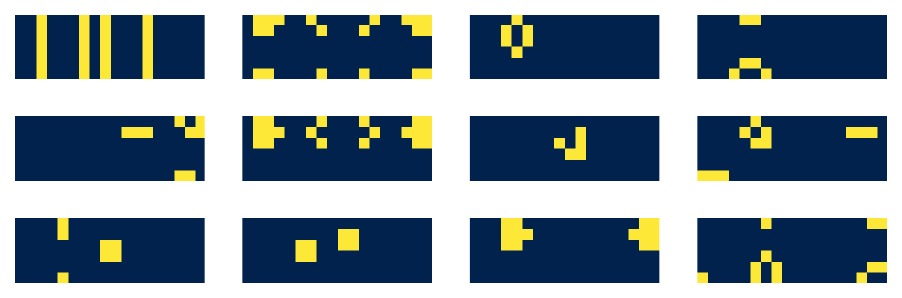}
    \end{subfigure}
    
    \vspace{0.5cm}
    
    \begin{subfigure}[b]{0.7\textwidth}
        \centering
        \includegraphics[width=\textwidth]{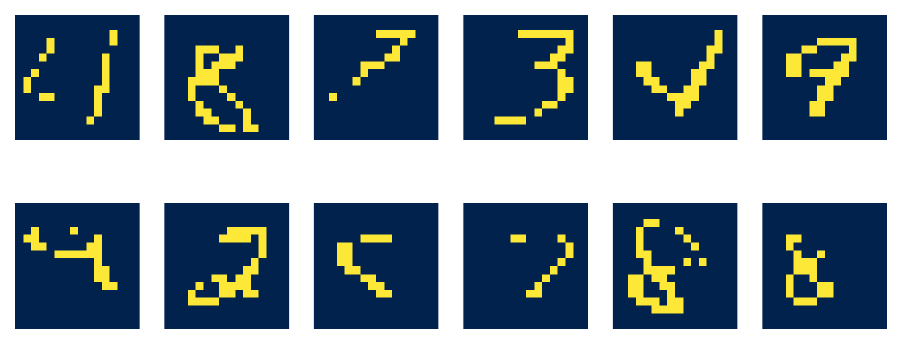}
    \end{subfigure}
    
    \caption{Examples from the cellular automaton (top) and USPS (bottom) training datasets.}
    \label{fig:training_examples}
\end{figure}

\printbibliography

\end{document}

%% file: header.tex
\usepackage[lmargin=2cm,rmargin=2cm]{geometry}
\usepackage[utf8]{inputenc}
\usepackage{amsmath,amssymb,amsfonts,mathtools,amsthm,stmaryrd,bbm,physics, dsfont}
\usepackage[group-separator={,}]{siunitx}
\usepackage{todonotes}
\usepackage{booktabs}
\usepackage{afterpage}
\usepackage{url}
\usepackage{hyperref}
\usepackage{comment}
\usepackage{bm}
\usepackage{subcaption}
\usepackage{graphicx}

\usepackage{algorithm}       % For the algorithm floating environment
\usepackage{algpseudocode}   % For the pseudocode layout

\theoremstyle{definition}

\newtheorem{theorem}{Theorem}

\newtheorem{definition}{Definition}

\usepackage{xspace}
\usepackage[title]{appendix}
\usepackage[capitalise]{cleveref}
\usepackage{authblk}

\usepackage{xcolor,color}

\newcommand{\gbbm}{GBBM}

\newcommand{\mmd}{$\rm{MMD}^2$}

\newcommand{\p}{p}

\usepackage[
    style=numeric-comp,
    sorting=none,
    defernumbers=true,
    maxbibnames=20,
    maxcitenames=20,
    date=year,
    isbn=false,
    giveninits=true,
    uniquename=init,
]{biblatex}

\renewbibmacro*{doi+eprint+url}{%
    \iftoggle{bbx:url}
    {\iffieldundef{doi}{\iffieldundef{eprint}{\usebibmacro{url+urldate}}{}}{}}
    {}%
    \newunit\newblock
    \iftoggle{bbx:eprint}
    {\iffieldundef{doi}{\usebibmacro{eprint}}{}}
    {}%
    \newunit\newblock
    \iftoggle{bbx:doi}
    {\printfield{doi}}
    {}
}

\renewbibmacro{in:}{}
\AtEveryBibitem{%
\clearfield{issue}%
\clearfield{number}}

\DefineBibliographyStrings{english}{%
    urlseen = {last accessed on},
}

%% file: references.bib
@article{Bulmer_2022_threshold,
   title={Threshold detection statistics of bosonic states},
   volume={106},
   ISSN={2469-9934},
   url={http://dx.doi.org/10.1103/PhysRevA.106.043712},
   DOI={10.1103/physreva.106.043712},
   number={4},
   journal={Phys. Rev. A},
   publisher={American Physical Society (APS)},
   author={Bulmer, J. F. F. and Paesani, S. and Chadwick, R. S. and Quesada, N.},
   year={2022},
   month=oct
}

@misc{kaposi2022polynomialspeeduptorontoniancalculation,
      title={Polynomial speedup in Torontonian calculation by a scalable recursive algorithm}, 
      author={Ágoston Kaposi and Zoltán Kolarovszki and Tamás Kozsik and Zoltán Zimborás and Péter Rakyta},
      year={2022},
      eprint={2109.04528},
      archivePrefix={arXiv},
      primaryClass={quant-ph},
      url={https://arxiv.org/abs/2109.04528}, 
}

@book{Caianiello1973,
  author    = {Caianiello, Eduardo R.},
  title     = {Combinatorics \& {R}enormalization in {Q}uantum {F}ield {T}heory},
  publisher = {Benjamin},
  address   = {Reading},
  volume    = {38},
  year      = {1973},
  url       = {https://www.osti.gov/biblio/4338754}
}

@article{Kruse_2019,
  title     = {Detailed study of {G}aussian boson sampling},
  author    = {Kruse, Regina and Hamilton, Craig S. and Sansoni, Linda and Barkhofen, Sonja and Silberhorn, Christine and Jex, Igor},
  journal   = {Phys. Rev. A},
  volume    = {100},
  issue     = {3},
  pages     = {032326},
  numpages  = {15},
  year      = {2019},
  month     = Sep,
  publisher = {American Physical Society},
  doi       = {10.1103/PhysRevA.100.032326},
  url       = {https://link.aps.org/doi/10.1103/PhysRevA.100.032326}
}

@article{hamilton2017,
  title     = {Gaussian {B}oson {S}ampling},
  author    = {Hamilton, C. S. and Kruse, R. and Sansoni, L. and Barkhofen, S. and Silberhorn, C. and Jex, I.},
  journal   = {Phys. Rev. Lett.},
  volume    = {119},
  issue     = {17},
  pages     = {170501},
  numpages  = {5},
  year      = {2017},
  publisher = {American Physical Society},
  doi       = {10.1103/PhysRevLett.119.170501},
  url       = {https://link.aps.org/doi/10.1103/PhysRevLett.119.170501}
}

@misc{bakó2025fermionicbornmachinesclassical,
      title={Fermionic Born Machines: Classical training of quantum generative models based on Fermion Sampling}, 
      author={Bence Bakó and Zoltán Kolarovszki and Zoltán Zimborás},
      year={2025},
      eprint={2511.13844},
      archivePrefix={arXiv},
      primaryClass={quant-ph},
      url={https://arxiv.org/abs/2511.13844}, 
}

@article{rudolph2024trainability,
   title={Trainability barriers and opportunities in quantum generative modeling},
   volume={10},
   ISSN={2056-6387},
   url={http://dx.doi.org/10.1038/s41534-024-00902-0},
   DOI={10.1038/s41534-024-00902-0},
   number={1},
   journal={npj Quantum Inf.},
   publisher={Springer Science and Business Media LLC},
   author={Rudolph, Manuel S. and Lerch, Sacha and Thanasilp, Supanut and Kiss, Oriel and Shaya, Oxana and Vallecorsa, Sofia and Grossi, Michele and Holmes, Zoë},
   year={2024},
   month=nov }

@article{kasture2023protocols,
   title={Protocols for classically training quantum generative models on probability distributions},
   volume={108},
   ISSN={2469-9934},
   url={http://dx.doi.org/10.1103/PhysRevA.108.042406},
   DOI={10.1103/physreva.108.042406},
   number={4},
   journal={Phys. Rev. A},
   publisher={American Physical Society (APS)},
   author={Kasture, Sachin and Kyriienko, Oleksandr and Elfving, Vincent E.},
   year={2023},
   month=oct }

@misc{recio2025iqpopt,
      title={{IQPopt: Fast optimization of instantaneous quantum polynomial circuits in JAX}}, 
      author={Erik Recio-Armengol and Joseph Bowles},
      year={2025},
      eprint={2501.04776},
      archivePrefix={arXiv},
      primaryClass={quant-ph},
      url={https://arxiv.org/abs/2501.04776}, 
}

@misc{recio2025train,
      title={Train on classical, deploy on quantum: scaling generative quantum machine learning to a thousand qubits}, 
      author={Erik Recio-Armengol and Shahnawaz Ahmed and Joseph Bowles},
      year={2025},
      eprint={2503.02934},
      archivePrefix={arXiv},
      primaryClass={quant-ph},
      url={https://arxiv.org/abs/2503.02934}, 
}

@article{oszmaniec2022fermion,
   title={{Fermion Sampling: A Robust Quantum Computational Advantage Scheme Using Fermionic Linear Optics and Magic Input States}},
   volume={3},
   ISSN={2691-3399},
   url={http://dx.doi.org/10.1103/PRXQuantum.3.020328},
   DOI={10.1103/prxquantum.3.020328},
   number={2},
   journal={PRX Quantum},
   publisher={American Physical Society (APS)},
   author={Oszmaniec, Michał and Dangniam, Ninnat and Morales, Mauro E.S. and Zimborás, Zoltán},
   year={2022},
   month=may }

@misc{bako2024problem,
      title={{Problem-informed Graphical Quantum Generative Learning}}, 
      author={Bence Bakó and Dániel T. R. Nagy and Péter Hága and Zsófia Kallus and Zoltán Zimborás},
      year={2025},
      eprint={2405.14072},
      archivePrefix={arXiv},
      primaryClass={quant-ph},
      url={https://arxiv.org/abs/2405.14072}, 
}

@article{Ankan2024,
  author  = {Ankur Ankan and Johannes Textor},
  title   = {pgmpy: A Python Toolkit for Bayesian Networks},
  journal = {J. Mach. Learn. Res.},
  year    = {2024},
  volume  = {25},
  number  = {265},
  pages   = {1--8},
  url     = {http://jmlr.org/papers/v25/23-0487.html}
}

@article{chow1968approximating,
  title={Approximating discrete probability distributions with dependence trees},
  author={Chow, CKCN and Liu, Cong},
  journal={IEEE Trans. Inf. Theory},
  volume={14},
  number={3},
  pages={462--467},
  year={1968},
  publisher={IEEE},
  doi={10.1109/TIT.1968.1054142}
}

@book{adamatzky2010game,
  title={Game of life cellular automata},
  author={Adamatzky, Andrew},
  volume={1},
  year={2010},
  publisher={Springer},
  doi={10.1007/978-1-84996-217-9}
}

@misc{kurkin2025universality,
      title={Universality and kernel-adaptive training for classically trained, quantum-deployed generative models}, 
      author={Andrii Kurkin and Kevin Shen and Susanne Pielawa and Hao Wang and Vedran Dunjko},
      year={2025},
      eprint={2510.08476},
      archivePrefix={arXiv},
      primaryClass={quant-ph},
      url={https://arxiv.org/abs/2510.08476}, 
}

@misc{garreau2017large,
      title={Large sample analysis of the median heuristic}, 
      author={Damien Garreau and Wittawat Jitkrittum and Motonobu Kanagawa},
      year={2018},
      eprint={1707.07269},
      archivePrefix={arXiv},
      primaryClass={math.ST},
      url={https://arxiv.org/abs/1707.07269}, 
}

@book{murphy2012machine,
  title={Machine learning: a probabilistic perspective},
  author={Murphy, Kevin P},
  year={2012},
  publisher={MIT press},
  url={https://dl.acm.org/doi/book/10.5555/2380985}
}

@article{ackley1985learning,
title = {A learning algorithm for boltzmann machines},
journal = {Cognitive Science},
volume = {9},
number = {1},
pages = {147-169},
year = {1985},
issn = {0364-0213},
doi = {10.1016/S0364-0213(85)80012-4},
url = {https://www.sciencedirect.com/science/article/pii/S0364021385800124},
author = {David H. Ackley and Geoffrey E. Hinton and Terrence J. Sejnowski}
}

@article{baydin2018automatic,
  title={Automatic differentiation in machine learning: a survey},
  author={Baydin, Atilim Gunes and Pearlmutter, Barak A and Radul, Alexey Andreyevich and Siskind, Jeffrey Mark},
  journal={J. Mach. Learn. Res.},
  volume={18},
  number={153},
  pages={1--43},
  doi={10.5555/3122009.3242010},
  year={2018}
}

@book{Goodfellow-et-al-2016,
    title={Deep Learning},
    author={Ian Goodfellow and Yoshua Bengio and Aaron Courville},
    publisher={MIT Press},
    note={\url{http://www.deeplearningbook.org}},
    year={2016}
}

@inproceedings{bengio2014deep,
  title={Deep generative stochastic networks trainable by backprop},
  author={Bengio, Yoshua and Laufer, Eric and Alain, Guillaume and Yosinski, Jason},
  booktitle={International conference on machine learning},
  pages={226--234},
  year={2014},
  organization={PMLR},
  doi={10.5555/3044805.3044918}
}

@article{mcclean2018barren,
  title={Barren plateaus in quantum neural network training landscapes},
  author={McClean, Jarrod R and Boixo, Sergio and Smelyanskiy, Vadim N and Babbush, Ryan and Neven, Hartmut},
  journal={Nat. Commun.},
  volume={9},
  number={1},
  pages={4812},
  year={2018},
  doi={10.1038/s41467-018-07090-4},
  publisher={Nature Publishing Group UK London}
}

@article{larocca2025barren,
  title={Barren plateaus in variational quantum computing},
  author={Larocca, Martin and Thanasilp, Supanut and Wang, Samson and Sharma, Kunal and Biamonte, Jacob and Coles, Patrick J and Cincio, Lukasz and McClean, Jarrod R and Holmes, Zo{\"e} and Cerezo, Marco},
  journal={Nat. Rev. Phys.},
  pages={1--16},
  year={2025},
  publisher={Nature Publishing Group UK London},
  doi={10.1038/s42254-025-00813-9}
}

@article{anschuetz2022quantum,
  title={Quantum variational algorithms are swamped with traps},
  author={Anschuetz, Eric R and Kiani, Bobak T},
  journal={Nat. Commun.},
  volume={13},
  number={1},
  pages={7760},
  year={2022},
  publisher={Nature Publishing Group UK London},
  doi={10.1038/s41467-022-35364-5}
}

@inproceedings{10.5555/3666122.3668062,
author = {Abbas, Amira and King, Robbie and Huang, Hsin-Yuan and Huggins, William J. and Movassagh, Ramis and Gilboa, Dar and McClean, Jarrod R.},
title = {On quantum backpropagation, information reuse, and cheating measurement collapse},
year = {2023},
publisher = {Curran Associates Inc.},
address = {Red Hook, NY, USA},
booktitle = {Proceedings of the 37th International Conference on Neural Information Processing Systems},
articleno = {1940},
numpages = {28},
location = {New Orleans, LA, USA},
series = {NIPS '23},
url={https://dl.acm.org/doi/10.5555/3666122.3668062}
}

@inproceedings{gilyen2019optimizing,
  title={Optimizing quantum optimization algorithms via faster quantum gradient computation},
  author={Gily{\'e}n, Andr{\'a}s and Arunachalam, Srinivasan and Wiebe, Nathan},
   url={http://dx.doi.org/10.1137/1.9781611975482.87},
   DOI={10.1137/1.9781611975482.87},
  booktitle={Proceedings of the Thirtieth Annual ACM-SIAM Symposium on Discrete Algorithms},
  pages={1425--1444},
  year={2019},
  organization={SIAM}
}

@article{bowles2025backpropagation,
  title={Backpropagation scaling in parameterised quantum circuits},
  author={Bowles, Joseph and Wierichs, David and Park, Chae-Yeun},
  journal={Quantum},
   url={http://dx.doi.org/10.22331/q-2025-10-02-1873},
   DOI={10.22331/q-2025-10-02-1873},
  volume={9},
  pages={1873},
  year={2025},
  publisher={Verein zur F{\"o}rderung des Open Access Publizierens in den Quantenwissenschaften}
}

@article{coyle2025training,
  title={Training-efficient density quantum machine learning},
  author={Coyle, Brian and Raj, Snehal and Mathur, Natansh and Cherrat, El Amine and Jain, Nishant and Kazdaghli, Skander and Kerenidis, Iordanis},
  journal={npj Quantum Information},
  volume={11},
  number={1},
  pages={172},
  year={2025},
  publisher={Nature Publishing Group UK London},
  doi={10.1038/s41534-025-01099-6}
}

@article{PRXQuantum.5.020328,
  title = {{Theory for Equivariant Quantum Neural Networks}},
  author = {Nguyen, Quynh T. and Schatzki, Louis and Braccia, Paolo and Ragone, Michael and Coles, Patrick J. and Sauvage, Fr\'ed\'eric and Larocca, Mart\'{\i}n and Cerezo, M.},
  journal = {PRX Quantum},
  volume = {5},
  issue = {2},
  pages = {020328},
  numpages = {40},
  year = {2024},
  month = {May},
  publisher = {American Physical Society},
  doi = {10.1103/PRXQuantum.5.020328},
  url = {https://link.aps.org/doi/10.1103/PRXQuantum.5.020328}
}

@article{Monbroussou2025trainability,
  doi = {10.22331/q-2025-05-15-1745},
  url = {https://doi.org/10.22331/q-2025-05-15-1745},
  title = {Trainability and {E}xpressivity of {H}amming-{W}eight {P}reserving {Q}uantum {C}ircuits for {M}achine {L}earning},
  author = {Monbroussou, L{\'{e}}o and Mamon, Eliott Z. and Landman, Jonas and Grilo, Alex B. and Kukla, Romain and Kashefi, Elham},
  journal = {{Quantum}},
  issn = {2521-327X},
  publisher = {{Verein zur F{\"{o}}rderung des Open Access Publizierens in den Quantenwissenschaften}},
  volume = {9},
  pages = {1745},
  month = may,
  year = {2025}
}

@article{PRXQuantum.3.030341,
  title = {Group-Invariant Quantum Machine Learning},
  author = {Larocca, Mart\'{\i}n and Sauvage, Fr\'ed\'eric and Sbahi, Faris M. and Verdon, Guillaume and Coles, Patrick J. and Cerezo, M.},
  journal = {PRX Quantum},
  volume = {3},
  issue = {3},
  pages = {030341},
  numpages = {25},
  year = {2022},
  month = {Sep},
  publisher = {American Physical Society},
  doi = {10.1103/PRXQuantum.3.030341},
  url = {https://link.aps.org/doi/10.1103/PRXQuantum.3.030341}
}

@article{PRXQuantum.4.010328,
  title = {{Exploiting Symmetry in Variational Quantum Machine Learning}},
  author = {Meyer, Johannes Jakob and Mularski, Marian and Gil-Fuster, Elies and Mele, Antonio Anna and Arzani, Francesco and Wilms, Alissa and Eisert, Jens},
  journal = {PRX Quantum},
  volume = {4},
  issue = {1},
  pages = {010328},
  numpages = {27},
  year = {2023},
  month = {Mar},
  publisher = {American Physical Society},
  doi = {10.1103/PRXQuantum.4.010328},
  url = {https://link.aps.org/doi/10.1103/PRXQuantum.4.010328}
}

@article{PRXQuantum.4.020327,
  title = {Speeding Up Learning Quantum States Through Group Equivariant Convolutional Quantum Ans\"atze},
  author = {Zheng, Han and Li, Zimu and Liu, Junyu and Strelchuk, Sergii and Kondor, Risi},
  journal = {PRX Quantum},
  volume = {4},
  issue = {2},
  pages = {020327},
  numpages = {18},
  year = {2023},
  month = {May},
  publisher = {American Physical Society},
  doi = {10.1103/PRXQuantum.4.020327},
  url = {https://link.aps.org/doi/10.1103/PRXQuantum.4.020327}
}

@article{benedetti2019generative,
  title={A generative modeling approach for benchmarking and training shallow quantum circuits},
  author={Benedetti, Marcello and Garcia-Pintos, Delfina and Perdomo, Oscar and Leyton-Ortega, Vicente and Nam, Yunseong and Perdomo-Ortiz, Alejandro},
  journal={npj Quantum Inf.},
   url={http://dx.doi.org/10.1038/s41534-019-0157-8},
   DOI={10.1038/s41534-019-0157-8},
  volume={5},
  number={1},
  pages={45},
  year={2019},
  publisher={Nature Publishing Group UK London}
}

@article{PhysRevA.98.062324,
  title = {Differentiable learning of quantum circuit Born machines},
  author = {Liu, Jin-Guo and Wang, Lei},
  journal = {Phys. Rev. A},
  volume = {98},
  issue = {6},
  pages = {062324},
  numpages = {9},
  year = {2018},
  month = {Dec},
  publisher = {American Physical Society},
  doi = {10.1103/PhysRevA.98.062324},
  url = {https://link.aps.org/doi/10.1103/PhysRevA.98.062324}
}

@article{coyle2020born,
  title={{The Born supremacy: quantum advantage and training of an Ising Born machine}},
  author={Coyle, Brian and Mills, Daniel and Danos, Vincent and Kashefi, Elham},
  journal={npj Quantum Inf.},
   url={http://dx.doi.org/10.1038/s41534-020-00288-9},
   DOI={10.1038/s41534-020-00288-9},
  volume={6},
  number={1},
  pages={60},
  year={2020},
  publisher={Nature Publishing Group UK London}
}

@article{scikit-learn,
  title={Scikit-learn: Machine Learning in {P}ython},
  author={Pedregosa, F. and Varoquaux, G. and Gramfort, A. and Michel, V.
          and Thirion, B. and Grisel, O. and Blondel, M. and Prettenhofer, P.
          and Weiss, R. and Dubourg, V. and Vanderplas, J. and Passos, A. and
          Cournapeau, D. and Brucher, M. and Perrot, M. and Duchesnay, E.},
  journal={Journal of Machine Learning Research},
  volume={12},
  pages={2825--2830},
  year={2011},
  url={https://jmlr.org/papers/v12/pedregosa11a.html}
}

@inproceedings{adam2014method,
  author    = {Kingma, D. and Ba, J.},
  booktitle = {International Conference on Learning Representations (ICLR)},
  title     = {Adam: A Method for Stochastic Optimization},
  year      = {2015},
  address   = {San Diego, CA, USA},
  url       = {https://inspirehep.net/literature/1670744}
}

@misc{jax2018github,
  author  = {James Bradbury and Roy Frostig and Peter Hawkins and Matthew James Johnson and Chris Leary and Dougal Maclaurin and George Necula and Adam Paszke and Jake Vander{P}las and Skye Wanderman-{M}ilne and Qiao Zhang},
  title   = {{JAX}: composable transformations of {P}ython+{N}um{P}y programs},
  url     = {http://github.com/google/jax},
  version = {0.3.13},
  journal = {GitHub},
  year    = {2018},
  urldate = {2025-05-17}
}

@online{qml_benchmarks,
  title   = {Benchmarking for quantum machine learning models},
  url     = {https://github.com/XanaduAI/qml-benchmarks},
  urldate = {2025-11-16}
}

@article{Grier_2022,
   title={{The Complexity of Bipartite Gaussian Boson Sampling}},
   volume={6},
   ISSN={2521-327X},
   url={http://dx.doi.org/10.22331/q-2022-11-28-863},
   DOI={10.22331/q-2022-11-28-863},
   journal={Quantum},
   publisher={Verein zur Forderung des Open Access Publizierens in den Quantenwissenschaften},
   author={Grier, Daniel and Brod, Daniel J. and Arrazola, Juan Miguel and Alonso, Marcos Benicio de Andrade and Quesada, Nicolás},
   year={2022},
   month=nov, pages={863} }

@article{li2025complexity_transition_dgbs,
  title   = {A complexity transition in displaced {Gaussian Boson Sampling}},
  author  = {Li, Zhenghao and Solomons, Naomi R. and Bulmer, Jacob F. F. and Patel, Raj B. and Walmsley, Ian A.},
  journal = {npj Quantum Information},
  volume  = {11},
  pages   = {119},
  year    = {2025},
  month   = jul,
  doi     = {10.1038/s41534-025-01062-5},
  url     = {https://doi.org/10.1038/s41534-025-01062-5},
  publisher = {Springer Nature}
}

@book{serafini,
  title     = {Gaussian {S}tates of {C}ontinuous {V}ariable {S}ystems},
  author    = {Serafini, Alessio},
  year      = {2017},
  publisher = {CRC Press, Taylor \& Francis Group}
}

@article{adesso:2014,
  title     = {Continuous {V}ariable {Q}uantum {I}nformation: {G}aussian {S}tates and {B}eyond},
  volume    = {21},
  issn      = {1793-7191},
  url       = {http://dx.doi.org/10.1142/S1230161214400010},
  doi       = {10.1142/s1230161214400010},
  number    = {01n02},
  journal   = {Open Systems \& Information Dynamics},
  publisher = {World Scientific Pub Co Pte Lt},
  author    = {G. Adesso and S. Ragy and A. R. Lee},
  year      = {2014},
  pages     = {1440001}
}

@misc{quesada2021quadratic,
  title         = {Quadratic speedup for simulating {G}aussian boson sampling},
  author        = {Nicolás Quesada and Juan Miguel Arrazola and Trevor Vincent and Haoyu Qi and Raúl García-Patrón},
  year          = {2021},
  eprint        = {2010.15595},
  archiveprefix = {arXiv},
  primaryclass  = {quant-ph}
}

@article{photonicadvantage1,
  title     = {Quantum computational advantage using photons},
  volume    = {370},
  issn      = {1095-9203},
  url       = {http://dx.doi.org/10.1126/science.abe8770},
  doi       = {10.1126/science.abe8770},
  number    = {6523},
  journal   = {Science},
  publisher = {American Association for the Advancement of Science (AAAS)},
  author    = {Zhong, H.-S. and
               Wang, H. and
               Deng, Y.-H. and
               Chen, M.-C. and
               Peng, L.-C. and
               Luo, Y.-H and
               Qin, J. and
               Wu, D. and
               Ding, X. and
               Hu, Y. and
               Hu, P. and
               Yang, X-Y and
               Zhang, W-J and
               Li, H. and
               Li et al., Y.},
  year      = {2020},
  pages     = {1460–1463}
}

@article{borealis,
  author  = {L. S. Madsen and
             F. Laudenbach and
             M. F. Askarani and
             F. Rortais and
             T. Vincent and
             J. F. F. Bulmer and
             F. M. Miatto and
             L. Neuhaus and
             L. G. Helt and
             M. J. Collins and
             A. Lita and
             T. Gerrits and
             S. Nam and
             V. Vaidya and
             M. Menotti et al.},
  year    = {2022},
  pages   = {75-81},
  title   = {Quantum computational advantage with a programmable photonic processor},
  volume  = {606},
  journal = {Nature},
  doi     = {10.1038/s41586-022-04725-x}
}

@misc{liu2025robustquantumcomputationaladvantage,
      title={Robust quantum computational advantage with programmable 3050-photon Gaussian boson sampling}, 
      author={Hua-Liang Liu and Hao Su and Si-Qiu Gong and Yi-Chao Gu and Hao-Yang Tang and Meng-Hao Jia and Qian Wei and Yukun Song and Dongzhou Wang and Mingyang Zheng and Faxi Chen and Libo Li and Siyu Ren and Xuezhi Zhu and Meihong Wang and Yaojian Chen and Yanfei Liu and Longsheng Song and Pengyu Yang and Junshi Chen and Hong An and Lei Zhang and Lin Gan and Guangwen Yang and Jia-Min Xu and Yu-Ming He and Hui Wang and Han-Sen Zhong and Ming-Cheng Chen and Xiao Jiang and Li Li and Nai-Le Liu and Yu-Hao Deng and Xiao-Long Su and Qiang Zhang and Chao-Yang Lu and Jian-Wei Pan},
      year={2025},
      eprint={2508.09092},
      archivePrefix={arXiv},
      primaryClass={quant-ph},
      url={https://arxiv.org/abs/2508.09092},
}

@article{zhong2021phase,
  title     = {Phase-{P}rogrammable {G}aussian {B}oson {S}ampling {U}sing {S}timulated {S}queezed {L}ight},
  volume    = {127},
  issn      = {1079-7114},
  url       = {http://dx.doi.org/10.1103/PhysRevLett.127.180502},
  doi       = {10.1103/physrevlett.127.180502},
  number    = {18},
  journal   = {Phys. Rev. Lett.},
  publisher = {American Physical Society (APS)},
  author    = {Zhong, H.-S. and Deng, Yu-Hao and Qin, Jian and Wang, Hui and Chen, Ming-Cheng and Peng, Li-Chao and Luo, Yi-Han and Wu, Dian and Gong, Si-Qiu and Su, Hao and Hu, Yi and Hu, Peng and Yang, Xiao-Yan and Zhang, Wei-Jun and Li, Hao and Li, Yuxuan and Jiang, Xiao and Gan, Lin and Yang, Guangwen and You, Lixing and Wang, Zhen and Li, Li and Liu, Nai-Le and Renema, Jelmer J. and Lu, Chao-Yang and Pan, Jian-Wei},
  year      = {2021}
}

@article{Bulmer_2022,
  doi       = {10.1126/sciadv.abl9236},
  year      = 2022,
  publisher = {American Association for the Advancement of Science ({AAAS})},
  volume    = {8},
  number    = {4},
  author    = {J. F. F. Bulmer and B. A. Bell and R. S. Chadwick and A. E. Jones and D. Moise and A. Rigazzi and J. Thorbecke and U.-U. Haus and T. Van Vaerenbergh and R. B. Patel and I. A. Walmsley and A. Laing},
  title     = {The boundary for quantum advantage in Gaussian boson sampling},
  journal   = {Science Advances}
}

@article{changhun_mps,
author = {Oh, Changhun and Liu, Minzhao and Alexeev, Yuri and Fefferman, Bill and Jiang, Liang},
year = {2024},
month = {06},
pages = {1461-1468},
title = {{Classical algorithm for simulating experimental Gaussian boson sampling}},
volume = {20},
journal = {Nature Physics},
doi = {10.1038/s41567-024-02535-8}
}

@misc{dodd2025fastfrugalgaussianboson,
      title={{A fast and frugal Gaussian Boson Sampling emulator}}, 
      author={Tom Dodd and Javier Martínez-Cifuentes and Oliver Thomson Brown and Nicolás Quesada and Raúl García-Patrón},
      year={2025},
      eprint={2511.14923},
      archivePrefix={arXiv},
      primaryClass={quant-ph},
      url={https://arxiv.org/abs/2511.14923}, 
}

@article{Salavrakos_2025,
   title={Error-mitigated photonic quantum circuit Born machine},
   volume={111},
   ISSN={2469-9934},
   url={http://dx.doi.org/10.1103/PhysRevA.111.L030401},
   DOI={10.1103/physreva.111.l030401},
   number={3},
   journal={Physical Review A},
   publisher={American Physical Society (APS)},
   author={Salavrakos, Alexia and Sedrakyan, Tigran and Mills, James and Mansfield, Shane and Mezher, Rawad},
   year={2025},
   month=mar }

@article{killoran2019continuous,
  title     = {Continuous-variable quantum neural networks},
  author    = {Killoran, Nathan and Bromley, Thomas R. and Arrazola, Juan Miguel and Schuld, Maria and Quesada, Nicol\'as and Lloyd, Seth},
  journal   = {Phys. Rev. Res.},
  volume    = {1},
  issue     = {3},
  pages     = {033063},
  numpages  = {22},
  year      = {2019},
  month     = Oct,
  publisher = {American Physical Society},
  doi       = {10.1103/PhysRevResearch.1.033063},
  url       = {https://link.aps.org/doi/10.1103/PhysRevResearch.1.033063}
}

@article{clements2017optimal,
  author    = {W. R. Clements and P. C. Humphreys and B. J. Metcalf and W. S. Kolthammer and I. A. Walmsley},
  journal   = {Optica},
  number    = {12},
  pages     = {1460--1465},
  publisher = {Optica Publishing Group},
  title     = {Optimal design for universal multiport interferometers},
  volume    = {3},
  year      = {2016},
  url       = {https://opg.optica.org/optica/abstract.cfm?URI=optica-3-12-1460},
  doi       = {10.1364/OPTICA.3.001460}
}

@article{reck,
  title = {Experimental realization of any discrete unitary operator},
  author = {Reck, Michael and Zeilinger, Anton and Bernstein, Herbert J. and Bertani, Philip},
  journal = {Phys. Rev. Lett.},
  volume = {73},
  issue = {1},
  pages = {58--61},
  numpages = {0},
  year = {1994},
  month = {Jul},
  publisher = {American Physical Society},
  doi = {10.1103/PhysRevLett.73.58},
  url = {https://link.aps.org/doi/10.1103/PhysRevLett.73.58}
}

@article{cariolaro2016bloch,
  title     = {{Bloch-Messiah reduction of Gaussian unitaries by Takagi factorization}},
  author    = {Cariolaro, G. and Pierobon, G.},
  journal   = {Phys. Rev. A},
  volume    = {94},
  number    = {6},
  pages     = {062109},
  year      = {2016},
  publisher = {APS}
}

@article{wigner_function,
  title = {{On the Quantum Correction For Thermodynamic Equilibrium}},
  author = {Wigner, E.},
  journal = {Phys. Rev.},
  volume = {40},
  issue = {5},
  pages = {749--759},
  numpages = {0},
  year = {1932},
  month = {Jun},
  publisher = {American Physical Society},
  doi = {10.1103/PhysRev.40.749},
  url = {https://link.aps.org/doi/10.1103/PhysRev.40.749}
}

@misc{bowles2025genomic,
    title={Genomic},
    author={Joseph Bowles},
    howpublished={\url{https://pennylane.ai/datasets/other/genomic}},
    year={2025}
}

@article{yelmen2021creating,
  title={Creating artificial human genomes using generative neural networks},
  author={Yelmen, Burak and Decelle, Aur{\'e}lien and Ongaro, Linda and Marnetto, Davide and Tallec, Corentin and Montinaro, Francesco and Furtlehner, Cyril and Pagani, Luca and Jay, Flora},
  journal={PLoS genetics},
  volume={17},
  number={2},
  pages={e1009303},
  year={2021},
  publisher={Public Library of Science San Francisco, CA USA}
}

@article{Wierichs_2022,
   title={General parameter-shift rules for quantum gradients},
   volume={6},
   ISSN={2521-327X},
   url={http://dx.doi.org/10.22331/q-2022-03-30-677},
   DOI={10.22331/q-2022-03-30-677},
   journal={Quantum},
   publisher={Verein zur Forderung des Open Access Publizierens in den Quantenwissenschaften},
   author={Wierichs, David and Izaac, Josh and Wang, Cody and Lin, Cedric Yen-Yu},
   year={2022},
   month=mar, pages={677} }

@article{mitarai2018,
  title     = {Quantum circuit learning},
  volume    = {98},
  issn      = {2469-9934},
  url       = {http://dx.doi.org/10.1103/PhysRevA.98.032309},
  doi       = {10.1103/physreva.98.032309},
  number    = {3},
  journal   = {Phys. Rev. A},
  publisher = {American Physical Society (APS)},
  author    = {Mitarai, K. and Negoro, M. and Kitagawa, M. and Fujii, K.},
  year      = {2018}
}

@article{schuld2019,
  title     = {Evaluating analytic gradients on quantum hardware},
  volume    = {99},
  issn      = {2469-9934},
  url       = {http://dx.doi.org/10.1103/PhysRevA.99.032331},
  doi       = {10.1103/physreva.99.032331},
  number    = {3},
  journal   = {Phys. Rev. A},
  publisher = {American Physical Society (APS)},
  author    = {Schuld, M. and Bergholm, V. and Gogolin, C. and Izaac, J. and Killoran, N.},
  year      = {2019}
}

@misc{hoefler2023disentanglinghypepracticalityrealistically,
      title={Disentangling Hype from Practicality: On Realistically Achieving Quantum Advantage}, 
      author={Torsten Hoefler and Thomas Haener and Matthias Troyer},
      year={2023},
      eprint={2307.00523},
      archivePrefix={arXiv},
      primaryClass={quant-ph},
      url={https://arxiv.org/abs/2307.00523}, 
}

@article{Sweke_2020,
   title={Stochastic gradient descent for hybrid quantum-classical optimization},
   volume={4},
   ISSN={2521-327X},
   url={http://dx.doi.org/10.22331/q-2020-08-31-314},
   DOI={10.22331/q-2020-08-31-314},
   journal={Quantum},
   publisher={Verein zur Forderung des Open Access Publizierens in den Quantenwissenschaften},
   author={Sweke, Ryan and Wilde, Frederik and Meyer, Johannes and Schuld, Maria and Faehrmann, Paul K. and Meynard-Piganeau, Barthélémy and Eisert, Jens},
   year={2020},
   month=aug, pages={314} }

@article{Pappalardo_2025,
   title={Photonic parameter-shift rule: Enabling gradient computation for photonic quantum computers},
   volume={111},
   ISSN={2469-9934},
   url={http://dx.doi.org/10.1103/PhysRevA.111.032429},
   DOI={10.1103/physreva.111.032429},
   number={3},
   journal={Physical Review A},
   publisher={American Physical Society (APS)},
   author={Pappalardo, Axel and Emeriau, Pierre-Emmanuel and de Felice, Giovanni and Ventura, Brian and Jaunin, Hugo and Yeung, Richie and Coecke, Bob and Mansfield, Shane},
   year={2025},
   month=mar }

@misc{go2025sufficientconditionshardnesslossy,
      title={Sufficient conditions for hardness of lossy {G}aussian boson sampling}, 
      author={Byeongseon Go and Changhun Oh and Hyunseok Jeong},
      year={2025},
      eprint={2511.07853},
      archivePrefix={arXiv},
      primaryClass={quant-ph},
      url={https://arxiv.org/abs/2511.07853}, 
}

@inproceedings{aaronson2011bosonsampling,
  author    = {Scott Aaronson and Alex Arkhipov},
  title     = {The Computational Complexity of Linear Optics},
  booktitle = {Proceedings of the Forty-Third Annual {ACM} Symposium on Theory of Computing (STOC '11)},
  pages     = {333--342},
  publisher = {ACM},
  year      = {2011},
  doi       = {10.1145/1993636.1993682},
  eprint    = {1011.3245},
  archivePrefix = {arXiv},
  primaryClass  = {quant-ph},
  url       = {https://arxiv.org/abs/1011.3245}
}

@article{Stano_2022,
   title={Review of performance metrics of spin qubits in gated semiconducting nanostructures},
   volume={4},
   ISSN={2522-5820},
   url={http://dx.doi.org/10.1038/s42254-022-00484-w},
   DOI={10.1038/s42254-022-00484-w},
   number={10},
   journal={Nat. Rev. Phys.},
   publisher={Springer Science and Business Media LLC},
   author={Stano, Peter and Loss, Daniel},
   year={2022},
   month=aug, pages={672–688} }

@article{Arute_2019,
   title={Quantum supremacy using a programmable superconducting processor},
   volume={574},
   ISSN={1476-4687},
   url={http://dx.doi.org/10.1038/s41586-019-1666-5},
   DOI={10.1038/s41586-019-1666-5},
   number={7779},
   journal={Nature},
   publisher={Springer Science and Business Media LLC},
   author={Arute, Frank and Arya, Kunal and Babbush, Ryan and Bacon, Dave and Bardin, Joseph C. and Barends, Rami and Biswas, Rupak and Boixo, Sergio and Brandao, Fernando G. S. L. and Buell, David A. and Burkett, Brian and Chen, Yu and Chen, Zijun and Chiaro, Ben and Collins, Roberto and Courtney, William and Dunsworth, Andrew and Farhi, Edward and Foxen, Brooks and Fowler, Austin and Gidney, Craig and Giustina, Marissa and Graff, Rob and Guerin, Keith and Habegger, Steve and Harrigan, Matthew P. and Hartmann, Michael J. and Ho, Alan and Hoffmann, Markus and Huang, Trent and Humble, Travis S. and Isakov, Sergei V. and Jeffrey, Evan and Jiang, Zhang and Kafri, Dvir and Kechedzhi, Kostyantyn and Kelly, Julian and Klimov, Paul V. and Knysh, Sergey and Korotkov, Alexander and Kostritsa, Fedor and Landhuis, David and Lindmark, Mike and Lucero, Erik and Lyakh, Dmitry and Mandrà, Salvatore and McClean, Jarrod R. and McEwen, Matthew and Megrant, Anthony and Mi, Xiao and Michielsen, Kristel and Mohseni, Masoud and Mutus, Josh and Naaman, Ofer and Neeley, Matthew and Neill, Charles and Niu, Murphy Yuezhen and Ostby, Eric and Petukhov, Andre and Platt, John C. and Quintana, Chris and Rieffel, Eleanor G. and Roushan, Pedram and Rubin, Nicholas C. and Sank, Daniel and Satzinger, Kevin J. and Smelyanskiy, Vadim and Sung, Kevin J. and Trevithick, Matthew D. and Vainsencher, Amit and Villalonga, Benjamin and White, Theodore and Yao, Z. Jamie and Yeh, Ping and Zalcman, Adam and Neven, Hartmut and Martinis, John M.},
   year={2019},
   month=oct, pages={505–510} }

@misc{cogan2021deterministicsourceindistinguishablephotons,
      title={A deterministic source of indistinguishable photons in a cluster state}, 
      author={Dan Cogan and Zu-En Su and Oded Kenneth and David Gershoni},
      year={2021},
      eprint={2110.05908},
      archivePrefix={arXiv},
      primaryClass={quant-ph},
      url={https://arxiv.org/abs/2110.05908}, 
}

@article{rajeev_2022,
   title={Suppressing quantum errors by scaling a surface code logical qubit},
   volume={614},
   ISSN={1476-4687},
   url={http://dx.doi.org/10.1038/s41586-022-05434-1},
   DOI={10.1038/s41586-022-05434-1},
   number={7949},
   journal={Nature},
   publisher={Springer Science and Business Media LLC},
   author={Acharya, Rajeev and Aleiner, Igor and Allen, Richard and Andersen, Trond I. and Ansmann, Markus and Arute, Frank and Arya, Kunal and Asfaw, Abraham and Atalaya, Juan and Babbush, Ryan and Bacon, Dave and Bardin, Joseph C. and Basso, Joao and Bengtsson, Andreas and Boixo, Sergio and Bortoli, Gina and Bourassa, Alexandre and Bovaird, Jenna and Brill, Leon and Broughton, Michael and Buckley, Bob B. and Buell, David A. and Burger, Tim and Burkett, Brian and Bushnell, Nicholas and Chen, Yu and Chen, Zijun and Chiaro, Ben and Cogan, Josh and Collins, Roberto and Conner, Paul and Courtney, William and Crook, Alexander L. and Curtin, Ben and Debroy, Dripto M. and Del Toro Barba, Alexander and Demura, Sean and Dunsworth, Andrew and Eppens, Daniel and Erickson, Catherine and Faoro, Lara and Farhi, Edward and Fatemi, Reza and Flores Burgos, Leslie and Forati, Ebrahim and Fowler, Austin G. and Foxen, Brooks and Giang, William and Gidney, Craig and Gilboa, Dar and Giustina, Marissa and Grajales Dau, Alejandro and Gross, Jonathan A. and Habegger, Steve and Hamilton, Michael C. and Harrigan, Matthew P. and Harrington, Sean D. and Higgott, Oscar and Hilton, Jeremy and Hoffmann, Markus and Hong, Sabrina and Huang, Trent and Huff, Ashley and Huggins, William J. and Ioffe, Lev B. and Isakov, Sergei V. and Iveland, Justin and Jeffrey, Evan and Jiang, Zhang and Jones, Cody and Juhas, Pavol and Kafri, Dvir and Kechedzhi, Kostyantyn and Kelly, Julian and Khattar, Tanuj and Khezri, Mostafa and Kieferová, Mária and Kim, Seon and Kitaev, Alexei and Klimov, Paul V. and Klots, Andrey R. and Korotkov, Alexander N. and Kostritsa, Fedor and Kreikebaum, John Mark and Landhuis, David and Laptev, Pavel and Lau, Kim-Ming and Laws, Lily and Lee, Joonho and Lee, Kenny and Lester, Brian J. and Lill, Alexander and Liu, Wayne and Locharla, Aditya and Lucero, Erik and Malone, Fionn D. and Marshall, Jeffrey and Martin, Orion and McClean, Jarrod R. and McCourt, Trevor and McEwen, Matt and Megrant, Anthony and Meurer Costa, Bernardo and Mi, Xiao and Miao, Kevin C. and Mohseni, Masoud and Montazeri, Shirin and Morvan, Alexis and Mount, Emily and Mruczkiewicz, Wojciech and Naaman, Ofer and Neeley, Matthew and Neill, Charles and Nersisyan, Ani and Neven, Hartmut and Newman, Michael and Ng, Jiun How and Nguyen, Anthony and Nguyen, Murray and Niu, Murphy Yuezhen and O’Brien, Thomas E. and Opremcak, Alex and Platt, John and Petukhov, Andre and Potter, Rebecca and Pryadko, Leonid P. and Quintana, Chris and Roushan, Pedram and Rubin, Nicholas C. and Saei, Negar and Sank, Daniel and Sankaragomathi, Kannan and Satzinger, Kevin J. and Schurkus, Henry F. and Schuster, Christopher and Shearn, Michael J. and Shorter, Aaron and Shvarts, Vladimir and Skruzny, Jindra and Smelyanskiy, Vadim and Smith, W. Clarke and Sterling, George and Strain, Doug and Szalay, Marco and Torres, Alfredo and Vidal, Guifre and Villalonga, Benjamin and Vollgraff Heidweiller, Catherine and White, Theodore and Xing, Cheng and Yao, Z. Jamie and Yeh, Ping and Yoo, Juhwan and Young, Grayson and Zalcman, Adam and Zhang, Yaxing and Zhu, Ningfeng},
   year={2023},
   month=feb, pages={676–681} }

@article{hull1994database,
  title={A database for handwritten text recognition research},
  author={Hull, Jonathan J.},
  journal={IEEE Transactions on Pattern Analysis and Machine Intelligence},
  volume={16},
  number={5},
  pages={550--554},
  year={1994},
  publisher={IEEE},
  doi={10.1109/34.291440}
}

@article{morningstar2018deep,
  title={Deep learning the ising model near criticality},
  author={Morningstar, Alan and Melko, Roger G},
  journal={Journal of Machine Learning Research},
  volume={18},
  number={163},
  pages={1--17},
  year={2018}
}

@inproceedings{kolarovszki_cvbm,
  author={Kolarovszki, Zoltán and Nagy, Dániel T. R. and Zimborás, Zoltán},
  booktitle={IEEE Quantum Week, QCE 2024 Proceedings},
  title={{On the Learning Abilities of Photonic Continuous-Variable Born Machines}},
  year={2024},
  volume={01},
  number={},
  pages={750-756},
  doi={10.1109/QCE60285.2024.00094}
}

@article{cepaite2022,
  title     = {A continuous variable {B}orn machine},
  volume    = {4},
  issn      = {2524-4914},
  url       = {http://dx.doi.org/10.1007/s42484-022-00063-3},
  doi       = {10.1007/s42484-022-00063-3},
  number    = {1},
  journal   = {Quantum Machine Intelligence},
  publisher = {Springer Science and Business Media LLC},
  author    = {Čepaitė, Ieva and Coyle, Brian and Kashefi, Elham},
  year      = {2022},
  month     = mar
}

@article{hofheinz2009synthesizing,
  title={Synthesizing arbitrary quantum states in a superconducting resonator},
  author={Hofheinz, Max and Wang, H and Ansmann, Markus and Bialczak, Radoslaw C and Lucero, Erik and Neeley, Matthew and O'connell, AD and Sank, Daniel and Wenner, J and Martinis, John M and others},
  journal={Nature},
  volume={459},
  number={7246},
  pages={546--549},
  year={2009},
  publisher={Nature Publishing Group UK London}
}

@article{vlastakis2013deterministically,
  title={Deterministically encoding quantum information using 100-photon Schr{\"o}dinger cat states},
  author={Vlastakis, Brian and Kirchmair, Gerhard and Leghtas, Zaki and Nigg, Simon E and Frunzio, Luigi and Girvin, Steven M and Mirrahimi, Mazyar and Devoret, Michel H and Schoelkopf, Robert J},
  journal={Science},
  volume={342},
  number={6158},
  pages={607--610},
  year={2013},
  publisher={American Association for the Advancement of Science}
}

@article{Hadfield2009,
  author  = {Hadfield, Robert H.},
  title   = {Single-photon detectors for optical quantum information applications},
  journal = {Nature Photonics},
  year    = {2009},
  volume  = {3},
  number  = {12},
  pages   = {696--705},
  month   = {Dec},
  doi     = {10.1038/nphoton.2009.230},
  url     = {https://doi.org/10.1038/nphoton.2009.230},
  issn    = {1749-4893}
}

@article{cahillglauber69,
  title = {Density Operators and Quasiprobability Distributions},
  author = {Cahill, K. E. and Glauber, R. J.},
  journal = {Phys. Rev.},
  volume = {177},
  issue = {5},
  pages = {1882--1902},
  numpages = {0},
  year = {1969},
  month = {Jan},
  publisher = {American Physical Society},
  doi = {10.1103/PhysRev.177.1882},
  url = {https://link.aps.org/doi/10.1103/PhysRev.177.1882}
}

@article{threshold_gbs,
  title = {Gaussian boson sampling using threshold detectors},
  author = {Quesada, Nicol\'as and Arrazola, Juan Miguel and Killoran, Nathan},
  journal = {Phys. Rev. A},
  volume = {98},
  issue = {6},
  pages = {062322},
  numpages = {9},
  year = {2018},
  month = {Dec},
  publisher = {American Physical Society},
  doi = {10.1103/PhysRevA.98.062322},
  url = {https://link.aps.org/doi/10.1103/PhysRevA.98.062322}
}

@article{kurkin2026,
title={{Universality of Classically Trainable, Quantum-Deployed Boson-Sampling Generative Models}},
author={Kurkin, Andrii and Chabaud, Ulysse and Kolarovszki, Zolt\'an and Bak\'o, Bence and Zimbor\'as, Zolt\'an and Dunjko, Vedran},
year={2026}
}

@article{Gerry_2010,
   title={The parity operator in quantum optical metrology},
   volume={51},
   ISSN={1366-5812},
   url={http://dx.doi.org/10.1080/00107514.2010.509995},
   DOI={10.1080/00107514.2010.509995},
   number={6},
   journal={Contemporary Physics},
   publisher={Informa UK Limited},
   author={Gerry, Christopher C. and Mimih, Jihane},
   year={2010},
   month=nov, pages={497–511} }

@article{gretton,
  author  = {Arthur Gretton and Karsten M. Borgwardt and Malte J. Rasch and Bernhard Sch{{\"o}}lkopf and Alexander Smola},
  title   = {{A Kernel Two-Sample Test}},
  journal = {Journal of Machine Learning Research},
  year    = {2012},
  volume  = {13},
  number  = {25},
  pages   = {723-773},
  url     = {http://jmlr.org/papers/v13/gretton12a.html}
}

@ARTICLE{fast_wh,
  author={Fino and Algazi},
  journal={IEEE Transactions on Computers}, 
  title={Unified Matrix Treatment of the Fast Walsh-Hadamard Transform}, 
  year={1976},
  volume={C-25},
  number={11},
  pages={1142-1146},

  doi={10.1109/TC.1976.1674569}}

@misc{bouland2025complexitytheoreticfoundationsbosonsamplinglinear,
      title={Complexity-theoretic foundations of BosonSampling with a linear number of modes}, 
      author={Adam Bouland and Daniel Brod and Ishaun Datta and Bill Fefferman and Daniel Grier and Felipe Hernandez and Michal Oszmaniec},
      year={2025},
      eprint={2312.00286},
      archivePrefix={arXiv},
      primaryClass={quant-ph},
      url={https://arxiv.org/abs/2312.00286}, 
}

@online{gaussian_bosonic_born_machines_repository,
  title   = {{Source code for the classical training of GBBMs}},
  url     = {https://github.com/Budapest-Quantum-Computing-Group/gaussian_bosonic_born_machines},
  urldate = {2026-03-11}
}

@article{gottlieb2026efficient,
  title={Efficient training of photonic quantum generative models},
  author={Gottlieb, Felix and Mezher, Rawad and Ventura, Brian and Mansfield, Shane and Salavrakos, Alexia},
  year={2026}
}
